\documentclass[11pt]{article}
\usepackage[margin=1in]{geometry}

\usepackage{amsmath}
\usepackage{amsthm}
\usepackage{amsfonts}
\usepackage{amssymb}

\usepackage{mathtools}
\usepackage{verbatim}

\usepackage[T1]{fontenc}

\usepackage[ruled]{algorithm}
\usepackage[noend]{algpseudocode}

\usepackage{enumerate}

\usepackage{amsmath,amsthm}
\usepackage{amssymb}
\usepackage{amsfonts}
\usepackage{paralist}
\usepackage[ruled]{algorithm}
\usepackage[noend]{algpseudocode}
\usepackage{hyperref}
\usepackage{ dsfont }

\usepackage{graphicx}
\usepackage{caption}
\usepackage{subcaption}

\usepackage{epsfig}
\usepackage{enumerate}

\usepackage{chngcntr}
\usepackage{apptools}

\newcommand{\remove}[1]{}

\usepackage{color}

\newtheorem{definition}{Definition}
\newtheorem{theorem}{Theorem}
\newtheorem{corollary}{Corollary}
\newtheorem{lemma}{Lemma}

\newtheorem{observation}{Observation}

\newtheorem{proposition}{Proposition}

\newif\ifrobocza
\roboczatrue

\newif\ifbibtex
\bibtextrue

\ifrobocza
\newcommand{\tj}[1]{{\color{blue}{#1}}}
\else
\newcommand{\todo}[1]{}
\newcommand{\tj}[1]{#1}
\fi

\newcommand{\prob}{\text{Prob}}

\newcommand{\cL}{{\mathcal L}}

\DeclarePairedDelimiter\ceil{\lceil}{\rceil}

\newcommand{\level}{\textit{level}}
\newcommand{\Lev}{\textit{Lev}}
\newcommand{\executora}{{Executor Algorithm\ }}
\newcommand{\executoraNOS}{{Executor Algorithm}}

\newcommand{\levelled}{levelled\ }
\newcommand{\Levelled}{Levelled\ }

\newcommand{\parent}{\textit{p}}

\newcommand{\gt}{\textit{$2$-height respecting tree}}
\newcommand{\thrt}{{2-HRT }}
\newcommand{\thrtNOS}{{2-HRT}}

\newcommand{\fasta}{Fast Broadcast Algorithm\ }
\newcommand{\fastaNOS}{Fast Broadcast Algorithm}
\newcommand{\fastesta}{Express Broadcast Algorithm\ }
\newcommand{\fastestaNOS}{Express Broadcast Algorithm}

\newcommand{\DOM}{\mathbf{DOM}}
\newcommand{\INF}{\mathbf{INF}}
\newcommand{\UNINF}{\mathbf{UNINF}}

\newcommand{\commentt}[1]{}

\begin{document}

\def\thefootnote{\fnsymbol{footnote}}

\title{{\bf 

Optimal-Length Labeling Schemes for Fast Deterministic Communication in Radio Networks}}
%\title{{\bf Exact trade-offs between time and advice\\ for topology recognition}}

\author{Adam Ga\'{n}czorz\footnotemark[1]
	\and Tomasz Jurdzi\'{n}ski\footnotemark[1] 
	\and  Andrzej Pelc\footnotemark[2]
}

\footnotetext[1]{Institute of Computer Science, University of Wroc{\l}aw, Poland. 
	Emails: {\tt \{adam.ganczorz,tju\}@cs.uni.wroc.pl}. 
	Supported by the Polish National Science Centre 
	grant 2020/39/B/ST6/03288.}

\footnotetext[2]{
	D\'epartement d'informatique, Universit\'e du Qu\'ebec en Outaouais, Gatineau,
	Qu\'ebec J8X 3X7, Canada. {\tt pelc@uqo.ca}. Partially supported by NSERC discovery grant RGPIN-2024-03767
	and by the Research Chair in Distributed Computing at the
	Universit\'e du Qu\'ebec en Outaouais.}

\maketitle

\thispagestyle{empty}

\begin{abstract}
	We consider two fundamental communication tasks in arbitrary radio networks: broadcasting (information from one source has to reach all nodes) and gossiping (every node has a message and all messages have to reach all nodes). Nodes are assigned labels that are (not necessarily different) binary strings. Each node knows its own label and can use it as a parameter in the same deterministic algorithm. The length of a labeling scheme is the largest length of a label. The goal is to find labeling schemes of asymptotically optimal length for the above tasks, and to design fast deterministic distributed algorithms for each of them, using labels of optimal length.

 Our main result concerns broadcasting.
 We show the existence of a labeling scheme of constant length that supports broadcasting in time $O(D+\log^2 n)$, where $D$ is the diameter of the network and $n$ is the number of nodes. This broadcasting time is an improvement over the best currently known $O(D\log n + \log^2 n)$ time of broadcasting with constant-length labels, due to  Ellen and Gilbert (SPAA 2020). It also matches the optimal broadcasting time in radio networks of known topology. Hence, we show that appropriately chosen node labels of constant length permit to achieve, in a distributed way, the optimal centralized broadcasting time. This is, perhaps, the most surprising finding of this paper. We are able to obtain our result thanks to a novel methodological tool of propagating information in radio networks, that we call a 2-height respecting tree.
 
	Next, we apply our broadcasting algorithm to solve the gossiping problem. 
    %on the length of the labels for this problem by Krisko and Miller (WG 2021).
	We get a gossiping algorithm working in time $O(D + \Delta\log n + \log^2 n)$, using a labeling scheme of optimal length $O(\log \Delta)$, where $\Delta$ is the maximum degree. Our time is the same as the best known gossiping time in radio networks of known topology.
	%Thus, our solutions for all three considered communication problems use labels of asymptotically optimal length, and distributed %algorithms using those labels work in the best known time, even when compared to algorithms for radio networks with known topology. %In the case of broadcasting and $k$-gathering, the time of our algorithms is provably optimal.

	\vspace*{0.5cm}
	
	\noindent
{\bf keywords: } radio network, distributed algorithms, algorithms with advice, labeling scheme, broadcasting, gossiping
	
\end{abstract}

\pagebreak

\renewcommand{\thefootnote}{\arabic{footnote}}

\section{Introduction}
\label{s:intro}

We consider two fundamental communication tasks often occurring in networks. In {\em broadcasting}, one node, called the {\em source}, has a message that must reach all other nodes. In gossiping, every node has a messasge and all messages have to reach all nodes.

We consider the above tasks in radio networks, modeled as undirected connected graphs. It is well known that, in the absence of labels, these communication tasks are infeasible in many networks, due to interferences. Hence we assume that nodes are assigned labels that are (not necessarily different) binary strings. Each node knows its own label and can use it as a parameter in the same deterministic algorithm. The length of a labeling scheme is the largest length of a label. The goal is to find labeling schemes of asymptotically optimal length for the above communication tasks, and to design fast deterministic distributed algorithms for each of them, using labels of optimal length.

\subsection{The model and the problem}

We consider radio networks modeled as simple undirected connected graphs. 
Throughout this paper, $G=(V,E)$ denotes the graph modeling the network, $n$ denotes the number of its nodes, $D$ its diameter, and $\Delta$ its maximum degree. 
At the cost of a small abuse of notation, we sometimes use $D$ to denote the height of a BFS spanning tree of a graph with a fixed root node. Note however that the height of a BFS tree is not larger than the diameter $D$ and not smaller than $D/2$, so the orders of magnitude are the same. In our probabilistic considerations concerning graphs with $n$ nodes, we use the term ``with high probability''
to mean ``with probability at least $1-1/n$''.

We use square brackets to indicate sets of consecutive integers: $[i,j] = \{i, \dots, j\}$ and $[i] = [1, i]$. 
All logarithms are to the base 2.
For simplicity of presentation, we assume
throughout the paper that the number of nodes of a graph $n$ is a power of 2, in order to
avoid rounding of logarithms. One can easily generalize all the results for arbitrary $n$, preserving asymptotic efficiency measures.

As usually assumed in the algorithmic literature on radio networks, nodes communicate in synchronous rounds (also called steps). All nodes start executing an algorithm in the same round. 
In each round, a node can either transmit a message to all its neighbors, or stay silent and listen. At the receiving end, a node $v$ hears the message from a neighbor $w$ in a given round, if $v$ listens in this round, and if $w$ is its only neighbor that transmits in this round. If more than one neighbor of a node $v$ transmits in a given round, there is a {\em collision} at $v$. Two scenarios concerning collisions were considered in the literature. The availability of {\em collision detection} means that node $v$ can distinguish collision from silence which occurs when no neighbor transmits. If collision detection is not available, node $v$ does not hear anything in case of a collision
(except the background noise that it also hears when no neighbor transmits). 
We do not assume collision detection. The time of a deterministic algorithm for a given task is the worst-case number of rounds it takes to solve it, expressed as a function of various network parameters. 

If nodes are indistinguishable (anonymous), i.e., in the absence of any labels, none of our communication problems can be solved, for example, in the four-cycle. Indeed, the node $w$ antipodal to $v$ cannot get the message of $v$ because, in any round, either both its neighbors transmit or both are silent. 
Hence we consider labeled networks, i.e., we assign binary strings, called {\em labels}, to nodes. A {\em labeling scheme} for a given network represented by a graph $G=(V,E)$ is any function $\cL$ from the set $V$ of nodes to the set $S$ of finite binary strings. The string $\cL(v)$ is called the label of the node $v$.
Note that labels assigned by a labeling scheme are not necessarily distinct. The {\em length} of a labeling scheme $\cL$ is the maximum length of any label assigned by it. Every node knows {a priori} only its label, and can use it as a parameter in the same deterministic algorithm

%If a labeling scheme for a graph can be constructed in time polynomial in the size $n$ of the graph, we say that this labeling scheme %is \emph{constructive}. Otherwise, the labeling scheme is \emph{non-constructive}.

%Our goal is to construct short labeling schemes for 

Solving distributed network problems with short labels can be seen in the framework of
algorithms with \emph{advice}.  In this paradigm that has recently got growing attention, an oracle knowing the network gives {advice} to nodes not knowing it,  in the form of binary strings, provided to nodes before the beginning of a computation. A distributed algorithm uses this advice to solve the problem efficiently.
The required size of advice (maximum length of the strings) can be considered a measure of the difficulty of the problem.
Two variations are studied in the literature: either the binary string given to nodes is the same for all of them \cite{DBLP:journals/talg/GlacetMP17} or different strings may be given to different nodes
\cite{DBLP:conf/spaa/EllenGMP19,EllenG20,DBLP:journals/mst/FraigniaudKL10,DBLP:journals/iandc/FuscoPP16}, as in the case
of the present paper. If strings may be different, they can be considered as labels assigned to nodes by a labeling scheme.
Such labeling schemes permitting to solve a given network task efficiently are also called {\em informative labeling schemes}.
One of the famous examples of using informative labeling schemes is to answer adjacency queries in graphs \cite{DBLP:journals/siamdm/AlstrupKTZ19}.

Several authors have studied the minimum amount of advice (i.e., label length) required to solve certain
network problems (see the subsection Related work). The framework of advice or of informative labeling schemes permits us to quantify the minimum amount of information used to solve a given network problem, regardless of the type of information that is provided.  It should be noticed that the scenario of the same advice given to all (otherwise anonymous) nodes would be useless in the case of radio networks: no deterministic communication could occur.

We now define formally our two communication tasks in a radio network $G=(V,E)$.\\
$\bullet$ {\bf broadcasting}\\
One node of the graph, called the {\em source}, has a broadcast message that has to reach all nodes $v\in V$.
A node which already knows the broadcast message is called an \emph{informed} node, otherwise the node is \emph{uninformed}. If a node $v$ receives the broadcast message for the first time in round $r$, from some neighbor $u$, we say that $u$ \emph{informed} $v$ in round $r$. An uninformed node $v$ is a \emph{frontier node} in a given round, if it is a neighbor of an informed node.
In our broadcasting algorithms, only informed nodes send messages.\\
%$\bullet$ {\bf gathering}\\
%One node of the graph is called the {\em sink}. Each node has a message, and all messages have to reach the sink.\\
$\bullet$ {\bf gossiping}\\
Each node $v\in V$ has a message, and all messages have to reach all nodes in $V$.

As it is customary in algorithmic literature concerning radio networks, we assume that when a node sends a message, this message can be of arbitrary size. In particular, a node could send its entire history (however, in our algorithms, messages will be usually shorter:
in broadcasting, some control messages will be appended to the source message, and in gossiping, all messages already known to a node will be combined in a single message).

Now our goal can be succinctly formulated as follows:
\begin{quotation}
For each of the above tasks, find an optimal-length labeling scheme permitting to accomplish this task, and design an optimal-time algorithm for this task, using a scheme of optimal length.\footnote{For the task of broadcasting, constant-length labeling schemes are known, so in this case the goal is to find a scheme of constant length supporting an optimal-time broadcasting algorithm.}
\end{quotation}

\subsection{Our results}
Our main result concerns broadcasting.
 We improve the best currently known time of deterministic broadcasting using labeling schemes of constant length, due to 
 Ellen and Gilbert (SPAA 2020) \cite{EllenG20}. As in \cite{EllenG20}, our results are of two types: {\em constructive}, where the labeling scheme used by the algorithm is explicitly constructed using an algorithm polynomial in $n$, and {\em non-constructive}, where we only prove the existence of the labeling scheme used by the algorithm, via the probabilistic method. The broadcasting algorithm from \cite{EllenG20} using a constructive constant-length labeling scheme runs in time $O(D\log^2n)$. We improve it to time $O\left(D+\min(D,\log n)\cdot \log^2n\right)$. The broadcasting algorithm from \cite{EllenG20} using a non-constructive constant-length labeling scheme runs in time $O(D\log n +\log^2n)$. We improve it to time $O(D+\log^2 n)$. This latter time is, in fact, the optimal deterministic broadcasting time in radio networks of known topology.\footnote{This means that every node has an isomorphic copy of the graph, with nodes labeled in the same way by unique identifiers, and a node knows its identifier. Deterministic algorithms using such knowledge are called {\em centralized}.}.
  Hence, we show that appropriately chosen node labels of constant length permit us to achieve, in a deterministic distributed way, the optimal centralized broadcasting time. This is, perhaps, the most surprising finding of this paper.
  We are able to obtain our result thanks to a novel methodological tool of propagating information in radio networks, that we call a 2-height respecting tree.

  Next, we apply our broadcasting algorithm to solve the gossiping problem. 
 Using the non-constructive version of our result for broadcasting, we get an algorithm for the gossiping problem, working in time $O(D + \Delta\log n + \log^2 n)$, that uses a (non-constructive) labeling scheme of optimal length $O(\log \Delta)$.\footnote{Using only constructive labeling schemes, the polylogarithmic summand in our complexity of gossiping changes from $\log^2n$ to $\min(D,\log n)\log^2n$.} 
 Our time is the same as the best {\em known} gossiping time for radio networks of known topology (without any extra assumptions on parameters), that follows from \cite{DBLP:journals/dc/GasieniecPX07}.
	
 %Thus, our solutions for all three considered communication tasks use labels of asymptotically optimal length, and distributed %algorithms using those labels work in the best known time, even when compared to algorithms for radio networks with known topology. In %the case of broadcasting and $k$-gathering, the time of our algorithms is provably optimal.

We summarize our results and compare them with previous most relevant results in Table~\ref{tab:results:whole:paper}.
\begin{table}[h]
\begin{center}
\begin{tabular}{|c|c|c|c|}
\hline
Ref. & Time & Length of labeling scheme & Constructive \\
\hline
\hline
\multicolumn{4}{|c|}{Broadcasting: centralized optimal\ time $O(D+\log^2n)$,  \cite{DBLP:journals/dc/GasieniecPX07,KP-DC-07}}\\
\hline
\hline
\cite{DBLP:conf/spaa/EllenGMP19} & $O(n)$ &  $2$ bits &  Yes \\
\cite{EllenG20} & $O(D\log n +\log^2n)$ & $3$ bits & No \\ 
\cite{EllenG20} & $O(D\log^2 n)$ & $3$ bits & Yes \\ 
\hline
here & $O(D+\log^2n)$ & $7$ bits & No \\
here & $O\left(D+\min(D,\log n)\cdot \log^2n\right)$ & $7$ bits & Yes \\
%\hline
%\hline
%\multicolumn{4}{|c|}{Gathering}\\
%\hline
%\hline
%\cite{DBLP:conf/wg/KriskoM21} & $O\left(n\min(k,\Delta^2))\right)$ & $\Theta\left(\min(\log k, \log\Delta)\right)$ & Yes \\
%here & $O(D+\Delta\log n+\log^2 n)$ & $\Theta(\log\Delta)$ & No \\
%here & $O(D+\Delta\log n+\log^3 n)$ & $\Theta(\log\Delta)$ & Yes \\
\hline
\hline
\multicolumn{4}{|c|}{Gossiping: centralized best time known $O(D+\Delta\log n+\log^2n)$, follows from \cite{DBLP:journals/dc/GasieniecPX07}}\\
\hline
\hline
here & $O(D+\Delta\log n+\log^2n)$ & $\Theta(\log \Delta)$ & No\\
here & {$O(D+\Delta\log n+\min(D,\log n)\log^2n)$} & $\Theta(\log \Delta)$ & Yes\\
\hline
\hline

\end{tabular}
\end{center}
\caption{Previous and our results.}
\label{tab:results:whole:paper}
\end{table}

\subsection{Related work}

The tasks of broadcasting and gossiping in radio networks were extensively investigated in algorithmic literature.
For deterministic algorithms, two important scenarios were studied. The first concerns centralized algorithms, in which each node knows the topology of the network and its location in it. Here, an optimal-time broadcasting algorithm was given in \cite{DBLP:journals/dc/GasieniecPX07,KP-DC-07} and the best known gossiping time (without any extra assumptions on parameters) follows from \cite{DBLP:journals/dc/GasieniecPX07}. For large values of $\Delta$, this was later improved in \cite{CMX}.  The second scenario
concerns distributed algorithms, where nodes have distinct labels, and every node knows its own label and an upper bound on the size of the network but does not know its topology. Here the best known broadcasting time that depends only on $n$ is $O(n\log n\log\log n)$ \cite{DeMarco}, later improved in \cite{CD} for some values of parameters $D$ and $\Delta$. For gossiping, the best known time  in arbitrary directed (strongly connected) graphs was given in \cite{GL,GRX} and the best known time for undirected graphs follows from \cite{Vaya}. Randomized distributed broadcasting was studied in \cite{KP,CR}, where optimal-time algorithms were obtained independently. For gossiping, optimal randomized time was given in \cite{GH}.

The advice paradigm has been applied to many different distributed network tasks: finding a minimum spanning tree  \cite{DBLP:journals/mst/FraigniaudKL10}, finding the topology of the network \cite{DBLP:journals/iandc/FuscoPP16}, and leader election \cite{DBLP:journals/talg/GlacetMP17}. In \cite{DBLP:conf/spaa/EllenGMP19} and \cite{EllenG20}, the task was broadcasting in radio networks, as in the present paper. In the above papers, advice was given to nodes of the network. Other authors considered the framework of advice for tasks executed by mobile agents navigating in networks, such as exploration \cite{GP} or rendezvous \cite{MP}. In this case, advice is given to mobile agents.

\subsection{Organization of the paper}
We present a high-level description of our results in Section~\ref{sec:highlevel}.
Section~\ref{sec:2high} contains basic notations and definitions concerning spanning trees. It also introduces the notion of a \gt\ (\thrtNOS) and the proof that one can built a BFS tree of each graph which is a \thrtNOS, the result essential for efficiency of our broadcasting algorithms.
In Section~\ref{sec:broadcasting}, we focus on broadcasting. We start with a modified variant of the \executora from \cite{EllenG20} and then gradually present components of our solutions, both non-constructive and constructive ones.
Section~\ref{sec:k:gathering} is devoted to the task of gossiping.
We introduce the auxiliary task of gathering, present a gathering algorithm using a labeling scheme of optimal length $O(\Delta)$ and
show that, by combining our algorithms for gathering and for broadcasting, one obtains a gossiping algorithm with optimal length of labels and with the best known time, even compared to algorithms for networks of known topology. 
Finally, in Section~\ref{sec:summary}, we conclude the paper and present some open problems.

\section{High-level Description of our Results}\label{sec:highlevel}
\subsection{High-level description of broadcasting}
\commentt{
Our final broadcasting algorithm with asymptotically optimal time complexity will use a levelled variant
of \executora from \cite{EllenG20}. As stated in \cite{EllenG20}, Theorem~12, their \executora accomplishes broadcasting in $O(D\log n+\log^2n)$ using non-construcive labeling scheme and in $O(D\log n+\log^2n)$ rounds with constructive labeling scheme. 
They also introduce the notion of a \levelled broadcast algorithm, although they do not use a \levelled variant of \executora in their paper. However,  their algorithm can be transformed into its ``\levelled'' variant preserving its asymptotic time complexity and constant size of labels. 

}
Our algorithms combine three mechanisms:
\begin{enumerate}
    \item The domination mechanism from \cite{DBLP:conf/spaa/EllenGMP19}.
    
    Computation is split into blocks of some constant number of rounds. At the beginning of  block $r$, a fixed  set $\DOM_r$ of nodes is active which is a minimal set of informed nodes
    with respect to inclusion that covers all frontier nodes. All elements of $\DOM_r$ simultaneously transmit in the first round of the block called the \textit{Broadcast} step. Minimality of $\DOM_r$ guarantees that each $v\in \DOM_r$ informs at least one  uninformed node. For each $v\in \DOM_r$, the labeling algorithm chooses exactly one such node $v'$ informed by $v$ in block $r$ as the feedback node of $v$ in that block.

Importantly, all feedback nodes can transmit simultaneously messages received by the nodes which serve as their witnesses.    These feedback nodes transmit in the second round of the block, called the \textit{Feedback} step. Their messages contain some information stored in their labels which instruct the corresponding nodes from $\DOM_r$ whether they
    should stay in $\DOM_{r+1}$ and instruct them about their behaviour in the remaining steps of the current block $r$.

    Nodes informed until block $r$ which are outside of $\DOM_r$ remain inactive to the end of an execution of the algorithm. The intuition regarding this 
    property is the fact that a node $v$ outside of $\DOM_r$ does not have its feedback node to instruct $v$ about its actions. On the other hand, $v$ cannot store this information in its own label for many blocks of computation, because it would require non-constant size of labels.

    As each block extends the set of informed nodes, the domination mechanism guarantees broadcasting in $O(n)$ time.
    
    \item The propagation mechanism from \cite{EllenG20}.

    In order to accelerate propagation of the broadcast message in the case when the diameter $D$ of the input graph is $o(n)$,
    ideas from a randomized seminal distributed algorithm of Bar-Yehuda et al.\ \cite{Bar-YehudaGI92} are applied.
    Namely, for appropriate random choices of informed nodes whether to transmit in a particular round, one can assure
    that the broadcast message is passed to the consecutive level of a BFS tree rooted at the source node $s$ in $O(\log n)$ rounds in expectation.
    This in turn gives randomized broadcasting in $O(D\log n+\log^2n)$ rounds with high probability.

    These random choices of nodes are mimicked in the labels of nodes. More precisely, the labels store some 0/1 random choices whether to transmit in a given block, assuring a given time bound. In particular, the feedback node of a node $v\in\DOM_r$ stores, in the bit \textit{Go} of its label, information whether $v$ should transmit. %in the \textit{Go} step of the current block. 
    Then, the nodes from $\DOM_r$ which received \textit{Go}=1 transmit the broadcast message in the separate \textit{Go} step of the block $r$. The labeling scheme obtained in this way is non-constructive.  Using %clever 
    ideas from \cite{ChlamtacK85} regarding centralized broadcasting in arbitrary bipartite graphs, one can obtain a constructive labeling scheme. However, the time of the broadcasting algorithm such a scheme would support becomes
    $O(D\log^2n)$ instead of time $O(D\log n+\log^2n)$ supported by the non-constructive scheme.

    \item The fast tracks mechanism.

    This mechanism is the main novelty of our solution and permits us to improve the broadcasting time from \cite{EllenG20}. The goal here is to implement ideas of a fast {\em centralized} algorithm into constant-size advice such that a {\em distributed} algorithm can somehow simulate the centralized one. The key ingredient of our approach is illustrated by the notion of 2-height respecting trees (\thrtNOS) and the fact that there exists a BFS tree which is also
    \thrtNOS, for each graph. The 2-height of a node $v$ in a tree intuitively denotes the maximum number of ``critical branches'' (causing large congestion) on a path from $v$ to a leaf.
    The maximum 2-height is always at most $\log n$.
    Each time the 2-height of a node $v$ and of some child $w$ of $v$ are the same, transmission of a message from $v$ to $w$ can be made in parallel
    with other similar transmissions from the level of $v$ dedicated to the particular value of $2$-height. Therefore, such an edge connecting  $v$ and $w$  with equal $2$-heights is called a \emph{fast edge}. As all but 
    %$\le \log n$ 
    $\log n$ 
    edges on each path from the root to a leaf are fast, the centralized algorithm from \cite{DBLP:journals/dc/GasieniecPX07} accomplishes broadcast in almost optimal time $O(D+\log^3 n)$.
    To this aim, the authors of \cite{DBLP:journals/dc/GasieniecPX07} make use of the notion of gathering trees which somehow minimize collisions between fast edges.
    Our notion of a \thrt imposes stronger requirements than gathering trees, making fast transmissions even more parallelizable.
    Then, the key challenge is an implementation of the idea of a centralized algorithm by constant-size labels instructing nodes of a distributed algorithm how to simulate the centralized algorithm. The main obstacle here comes from the domination mechanism which switches off some nodes irreversibly, preventing them from transmitting any message starting from the block $r$ in which they are outside of the minimal dominating set $\DOM_r$. We show that, for each such node, one can determine its ``rescue node'' still present in the dominating set, such that its transmission on behalf of a switched off node does not cause additional collisions.
    Here, the properties of a \thrt are essential.

    Our final solution using this mechanism gives a non-constructive labeling scheme of constant length, supporting broadcasting in time $O(D+\log^2n)$, which is optimal, even for centralized algorithms.
    Using the technique from \cite{ChlamtacK85} we can build labels constructively at the cost of increasing time complexity of broadcasting to $O\left(D+\min(D,\log n)\cdot \log^2n\right)$.

\end{enumerate}

\subsection{High-level description of gossiping}

In order to solve the gossiping problem, we introduce the auxiliary task of {\em gathering}: each node of the graph has a message, and all messages have to reach a designated node called the {\em sink}. We provide a gathering algorithm working in time $O(D+\Delta\log n+\log^2 n)$ and using a labeling scheme of length $O(\log\Delta)$.

To this aim, we again make use of properties of a \thrt to implement the centralized algorithm from \cite{DBLP:journals/dc/GasieniecPX07} in a distributed way, using short labels. 
Let $T$ be a BFS tree of the input graph which is also a \thrtNOS.
The centralized algorithm from \cite{DBLP:journals/dc/GasieniecPX07} determines the unique round $t(v)$ in which each node $v$ transmits all messages from its subtree of $T$\footnote{The authors of \cite{DBLP:journals/dc/GasieniecPX07} use the notion of gathering trees in their paper, but \thrt satisfy all properties of gathering trees as well.} to the parent of $v$. The value of $t(v)$ is chosen in such a way that the message is successfully received by the parent of $v$, i.e., there are no collisions at the parent of $v$.
These collision-free transmissions are assured by the properties of \emph{gathering trees} from \cite{DBLP:journals/dc/GasieniecPX07} which are also satisfied by \thrtNOS. The value of $t(v)$ depends on parameters $D$, $\level(v)$, $h_2(v)\in [0,\log n]$, $\Delta$ and on some auxiliary label $s(v)\in[0,\Delta-1]$. Thus, while $\Delta$ and $s(v)$ can be encoded in the label of $v$ using $O(\log \Delta)$ bits, we cannot store $D$, $h_2(v)$ and $\level(v)$ in the label if we want to get a labeling scheme of length $O(\log\Delta)$. To this aim we use the appropriately modified Size Learning Algorithm from \cite{DBLP:conf/wdag/GanczorzJLP21} followed by an acknowledged broadcasting algorithm to share information about the value of $D$ among all nodes and assure that nodes learn their levels by adding one to the values of levels of nodes from the preceding level, during an execution of the broadcasting algorithm. Finally, each leaf is marked as such by an appropriate bit of its label.
The fact that a node $v$ is a leaf implies also that $h_2(v)=0$. Other nodes learn their values of $h_2$ by modifying the maximal values of $h_2$ of their children when they receive all messages from them.
\commentt{
In order to perform $k$-gathering quickly, let us think that we try to transmit all messages to the sink node simultaneously, hoping that no collisions occur. Then, we could deliver all messages in at most $D$ rounds. 
In order to deal with prospective collisions, we will try to trade necessity of a slowdown 
in order to avoid collisions
with ability to combine more messages in one node. A rough idea is as follows. If transmissions from $u$ to $u'$ and from $v$ to $v'$ collide, i.e., there is an edge $(u,v')$ or $(v,u')$, we perform transmissions of $u$ and $v$ in separate rounds. W.l.o.g.\ assume that there is the edge $(u,v')$ in the communication graph.
In order to compensate for this slowdown, we combine messages collected up to the current round in $u$ and $v$ together, in the node $v'$. Then they are transmitted from $v'$ together as one packet on a single path from $v'$ to the root. This idea helps to build a $O(D+k)$ rounds $k$-gathering with $O(\log k)$ bits labels. However, in the case when $k\gg \Delta$, this size of labels is not the optimal $O\left(\min(\log k,\log \Delta)\right)$. In this case, we cannot encode the value of $k$ in labels while our algorithm requires knowledge of this value in order to coordinate distributed executions of the algorithm. In order to deal with this case, we switch between the pattern designed for the case $k\ge \Delta$ and a $2$-hop $(\Delta^2+1)$-coloring allowing for collision-free transmission in radio networks.
}

Our solution of the gossiping problem roughly works as follows. First, we gather all messages in an arbitrary node $s$ of the input graph, executing our gathering algorithm. Then, all messages collected at $s$ are distributed using
our broadcasting algorithm. An obstacle which arises in implementing this idea is caused by the fact that all nodes have to be coordinated so that they know when the consecutive subroutines of the final algorithm start. We overcome this difficulty by using an \emph{acknowledged} broadcasting algorithm.

%\section{Basic notions and tools}\label{sec:basic}
\section{$2$-height Respecting Trees}\label{sec:2high}

For a rooted tree $T$ with the root node $r$, we denote
%\begin{itemize}
	%\item 
the parent of a node $v\neq r$ as $\parent(v)$.
The \emph{level} of a node $v$ in the tree $T$ is equal to its distance to the root $r$.
The level of $v$ is denoted as $\level(v)$. Thus, in particular, $\level(r)=0$.
For a fixed graph $G=(V,E)$ and a node $r\in V$, the set of nodes at distance $l\ge 0$ from $r$ will be called the \emph{level $l$} and denoted as $L_l$. Thus, in particular, $L_0=\{r\}$, $L_1$ is the set of neighbors of $r$ and $L_l=$ $\{v\,|\, \level(v)=l\}$.

Now, we define the notion of the 2--height respecting tree (\thrtNOS), 
resembling gathering trees introduced in \cite{DBLP:journals/dc/GasieniecPX07}.
However, it is important to note that 2--height respecting trees must satisfy stronger properties than gathering trees.
That is, each \thrt is a gathering tree while a gathering tree might not be a \thrtNOS.

\begin{definition}[2-height, fast edge, slow edge]\label{def:2h}
	The \textit{2-height} of a node $v$ of a rooted tree $T$, denoted as $h_2(v)$, is defined as follows:
	\begin{itemize}
		\item If $v$ is a leaf then $h_2(v) = 0$.
		\item If $v$ is not a leaf, it has the children $u_1, u_2, \dots, u_k$ for $k\ge 1$ and there is exactly one node $u_i$ such that $h_2(u_i) = \max_j \{h_2(u_j)\}_{j\in[k]} $ then $h_2(v) = \max_j \{h_2(u_j)\}_{j\in[k]}$.
		\item If $v$ is not a leaf, it has the children $u_1, u_2, \dots, u_k$ for $k\ge 1$ and there are two or more nodes $u_i$ such that $h_2(u_i) = \max_j \{h_2(u_j)\}_{j\in[k]} $ then $h_2(v) = \max_j \{h_2(u_j)\}_{j\in[k]} +1$.
	\end{itemize}
If $h_2(v)=h_2(\parent(v))$ in a rooted tree $T$ then the edge $(p(v),v)$ in $T$ is called a \emph{fast edge}. 
Otherwise, the edge $(p(v),v)$ is a \emph{slow edge}.
\end{definition}

\commentt{\begin{definition}[2-height]
	The 2-height of a node $v$ is defined as follows
	\begin{itemize}
		\item If $v$ is a leaf node then $h_2(v)=0$
		\item If $v$ is not a leaf node and there is exactly one child $u$ of $v$ such that $h_2(u) = max\{h_2(a) | a \text{ is a child of } v\}$ then $h_2(v) = h_2(u)$
		\item If $v$ is not a leaf node and there is more than one such child then \\$h_2(v) = max\{h_2(a) | a \text{ is a child of } v\} + 1$	
	\end{itemize}
\end{definition}}

\begin{definition}[\gt]\label{def:2hrt} 
A rooted tree $T$ is a \gt\ (\thrt\ for short) of a graph $G$ if it is a BFS Tree of $G$ satisfying the following property:
\begin{enumerate} %[(a)]
\item[($\star$)]
\commentt{
For each two nodes $u,u'$ such that $\level(u)=\level(u')$, $\parent(u)\neq\parent(u')$ and
    $$h_2(u) = h_2(u') = h_2(\parent(u)) = h_2(\parent(u')),$$
there is no common neighbor $v$ of $u$ and $u'$ in $G$ such that $\level(v) = \level(u) -1 = \level(u')-1$.
}
For each two nodes $u,u'$ such that $\level(u)=\level(u')$ and
$$h_2(u) = h_2(u') = h_2(\parent(u)) = h_2(\parent(u')),$$
%then 
%either $\parent(u)=\parent(u')$ or
there is no common neighbor $v$ of $u$ and $u'$ in $G$ such that $\level(v) = \level(u) -1 = \level(u')-1$.
\end{enumerate}
\end{definition}
%Let us note here that the notion of \gt\ resembles gathering trees from \cite{DBLP:journals/dc/GasieniecPX07}. 
The key difference between gathering trees from \cite{DBLP:journals/dc/GasieniecPX07} 
and \gt \textit{s} is that the nodes $u$ and $u'$ on the same level:
\begin{itemize}
    \item cannot have the same parent, i.e., $\parent(u)\neq \parent(u')$ in the case of gathering trees,
    \item cannot even have a common neighbor on the level $\level(u)-1$ in the case of \thrtNOS.
\end{itemize}

The following key lemma  shows that, for each graph $G=(V,E)$ and each of its nodes $r\in V$, there exists a \thrtNOS\ of $G$ rooted in $r$. Moreover, such a tree can be constructed in polynomial time.
\begin{lemma}\label{lem:2hrt}
	For every graph $G=(V,E)$ and each $r\in V$, one can construct a BFS  %$gt$ 
 spanning tree $T$ of $G$ rooted at $r$ such that $T$ is a \thrtNOS. Moreover, the tree $T$ can be constructed in time $\text{poly}(n)$. 
	%Each BFS Tree can be changed into a \gt [in time $poly(n)$].
\end{lemma}

\begin{proof}
	Consider any BFS Tree $T$ of $G$ with the root equal to $r$. We will present the algorithm $A$ that modifies $T$ so that the resulting tree is a \thrtNOS\ of $G$ with the root $r$ and it is still a BFS tree. 
	The algorithm $A$ makes changes %into the \gt 
	``level by level'', starting from the largest level. $A$ first changes some edges between the level
	$l_{\max}=\max_{v\in V}\level{(v)}$ and the level $l_{\max}-1$, then it changes some edges connecting nodes from the level $l_{max}-1$ with the nodes from the level $l_{max}-2$ and so on, until the level $1$.
	
	Formally, the proof goes by induction on the decreasing order of levels. 
	Assume that, for some $l>0$, the property $(\star)$ is satisfied for all nodes on levels $l'>l$.
	%at nodes at one chosen level $l$.
	
We distinguish two types of violations of the ($\star$) property from Definition~\ref{def:2hrt}. Type 
 (a) violation holds if the common neighbor $v$ of $u$ and $u'$ violating the property is the parent of either $u$ or $v$ (cf. Fig. \ref{fig:2h:a}).
 Otherwise, we have type (b) violation (cf. Fig. \ref{fig:2h:b}).
	
If %both conditions \textit{(a)} and \textit{(b)} are satisfied 
	there is no violation of ($\star$) for all pairs $u, u'\in V$ from the level $l$,
	%for each pair of vertices on level $l$, 
	then the algorithm can go to the level $l-1$. 
	%If not the algorithm will proceed in the following way:
	Otherwise, the algorithm proceeds in the following way.
	
	First, we check if type (a) violation appears for any two nodes $u,u'$ from level $l$. 
	%violate the condition \textit{(a)}. 
	If this is the case for some $u\neq u'$ then $\parent(u) \neq \parent(u')$ and we change the parent of $u$ to be equal to $\parent(u')$. 
	This change will cause the increase of the value $h_2(\parent(u'))$ by one, since the maximal value of 2-heights of its children $u'$ becomes the 2-height of two of its children $u$ and $u'$ after the change.
	%both $u$ and $u'$ to have different value \textit{2-height} from their parent as $h_2(\parent(v))$ will %be now equal to $h_2(v)+1$. 
	Thus, eventually, $$h_2(u)=h_2(u')=h_2(\parent(u))-1=h_2(\parent(u'))-1,$$
	since the parent of $u$ and the parent of $u'$ are equal after our change.
	This relationship implies that $u$ and $u'$ no longer cause type (a) violation of the ($\star$) property -- see Figure~\ref{fig:2h:a}.
        \begin{figure}
        \centering
        \includegraphics[width=0.85\linewidth]{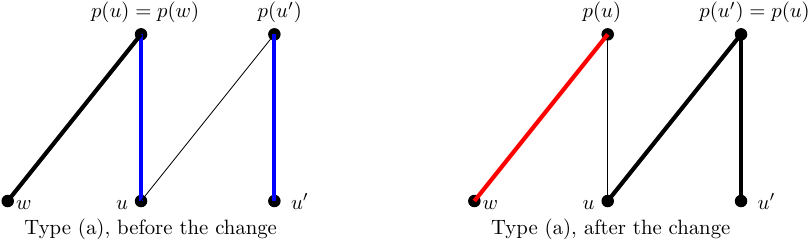}
        \caption{Illustration for type (a) violation. Fat edges connect nodes with their parents. Blue edges are fast, i.e., connect children with their parents such that the value of $h_2$ of the child and the parent are equal. Red edges might be fast, but it is not determined. Non-fat edges are not fast.}
        \label{fig:2h:a}
        \end{figure}        
 
	Now, assume that there is no type (a) violation $(\star)$ on the level $l$.
	Then, we check whether there are two nodes $u, u'$ on the level $l$ witnessing type (b) violation. That is, $u$ and $u'$ have a common neighbor $v$ on the level $\level(u)-1$ such that $v$ 
	is not the parent of $u$ or $u'$. 
	Then, we set $\parent(u) = \parent(u') = v$ and 
 %thus eliminating the violation of ($\star$) by the pair $u,u'$, since $h_2(v)$ which is the common parent of $u$ and $u'$ becomes larger than $h_2(u)=h_2(u')$ by the definition of the $2-height$ of a node. 
  thus the new parent $v$ of $u$, $u'$ has the new value of 2-height $h_2(v) > h_2(u)=h_2(u')$ and consequently $u$ and $u'$ no longer violate the property ($\star$)  -- see Fig. \ref{fig:2h:b}.
 After such elimination of type (b) violation, algorithm $A$ starts again checking whether there is type (a) violation and so on.
	
	Crucially, we can bound from above the number of such eliminations
 %steps 
 by proving monotonicity of the change of some specific measure of breach of the rule ($\star$) in the current tree $T$ with respect to the subgraph of $G$ limited to the edges connecting the nodes from the level $l$ with the nodes from the level $l-1$. 
	To this aim we define the function
	$$ S_{G}(T,l) = \sum_{\{u\,|\,\level(u)=l\text{ and }h_2(u)=h_2(\parent(u))\}} h_2(u).$$
	Below, we show that each change of edges of $T$ made by the algorithm $A$ eliminating some violation of the ($\star$) property decreases also the value of $S_{G}(T,l)$. %Let $S=S_{G,T,l}$
	
	First analyse an occurrence of type (a) violation. Let $u,u'$ be the ``violators'' and let $w$ be a child of $\parent(u)$ with the highest 2-height apart from $u$, if such $w$ exists -- see Figure~\ref{fig:2h:a}. As $h_2(u) = h_2(\parent(u))$, the value of $h_2(w)$ must be smaller then $h_2(u)$.  Then after changing $\parent(u)$ to the node $\parent(u')$, $u$ and $u'$ have now different 2-heights than their parent, since $h_2(u)=h_2(u')$ and $u,u'$ are siblings -- see the definition of 2-height.
	Thus, $2$-heights of $u$ and $u'$ no longer contribute to $S_{G}(,T,l)$. 
 On the other hand, it might be the case that $w$ has not contributed to $S_{G}(T,l)$ before the change since then $h_2(w)<h_2(\parent(w))=h_2(\parent(u))$. However it might happen that $h_2(w)=h_2(\parent(w))$ after the change. 
 Let $S$ be the value of $S_G(T,l)$ before the change and $S'$ the value of $S_G(T,l)$ after the change.
 Then,
 the values of $S$ and $S'$ %of $S_{G}(T,l)$ 
 satisfy the inequality
$$S'\leq S - h_2(u) - h_2(u') + h_2(w) \leq S - 1 - h_2(u') < S - 1,$$
since $h_2(w)<h_2(u)$.
	
	A similar reasoning can be applied to the case when the algorithm $A$ handles a type (b) violation. Let $u,u'$ be our ``violators'', and let $w,w'$ be theirs respective siblings in the tree $T$, with the greatest 2-height. We have $h_2(w) < h_2(u)$ and $h_2(w') < h_2(u')$. After replacing the parents of $u, u'$ with $v$, the nodes $u$ and $u'$ no longer contribute to $S_{G}(T,l)$, and the only possible new ``contributors'' to $S_{G}(T,l)$ are $w, w'$ -- see Figure~\ref{fig:2h:b}. 
         \begin{figure}
        \centering
        \includegraphics[width=0.95\linewidth]{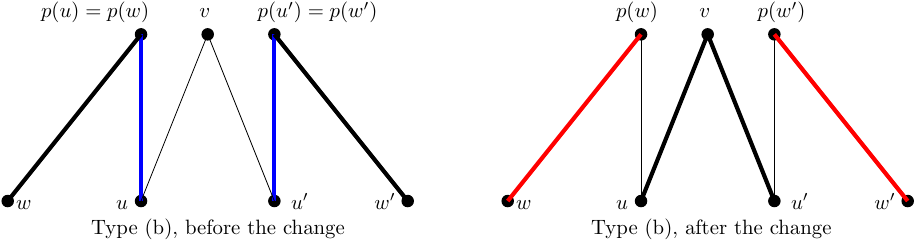}
        \caption{Illustration for type (b) violation. Fat edges connect nodes with their parents. Blue edges are fast, i.e., connect children with their parents such that the value of $h_2$ of the child and the parent are equal. Red edges might be fast, but it is not determined. Non-fat edges are not fast.}
        \label{fig:2h:b}
        \end{figure}        

	As before, let $S'$ be the value of $S_{G}(T,l)$, after replacement of edges fixing the type (b) violation, and let $S$ be the value of $S_G(T,l)$ before this replacement.
	Then the value of $S'$ can be at most
	$$S' \leq S - h_2(u) - h_2(u') + h_2(w) + h_2(w') \leq S - 1 -1 < S -2,$$
	since $h_2(w)<h_2(u)$ and $h_2(w')<h_2(u')$.
	
	The above inequalities show that each change of edges by the algorithm $A$ decreases the value of $S_{G}(T,l)$. As $S_{G}(T,l)$ cannot be negative and 2-height of every $v\in V$ is at most $\log n$ (Lemma~\ref{l:2h:max}, \cite{DBLP:journals/iandc/ChrobakCGK18}), we see that there are no violations of the ($\star$) property after
	at most $n \log n$ such changes made by $A$
    removing violations of the ($\star$) property. % changing the edges. 

	As changes of edges connecting nodes of level $l$ with nodes of level $l-1$ do not affect 2-heights and edges between nodes on larger levels, the algorithm $A$ can iteratively change levels starting from the greatest one. This shows that $A$ produces a \gt. As the number of changes at each level is polynomial and the largest level $l_{\max}$ is at most $n$, the algorithm $A$ works in time polynomial in $n$.
\end{proof}

\section{Broadcasting Algorithms}\label{sec:broadcasting}
This section is devoted to the broadcasting problem. 
As our solutions to this problem are build in the framework of the \executora introduced in \cite{EllenG20}, we start with the description of a modified variant of this algorithm. In Section~\ref{sec:lev:alg} we present and analyze the \Levelled \executoraNOS, which extends the original algorithm so that only nodes from specific levels can transmit and listen in particular rounds. Subsequently, this algorithm will be improved to get our final result.
\commentt{
Then, in Section~\ref{sec:alg:broadcast:our:constr} we add essential extensions to \levelled \executora which lead to non-constructive broadcasting algorithm with constant size labels and asymptotically optimal time. 
Then, in Section~\ref{?}, we define a class of so-called monotone domination resistant centralized broadcasting algorithms. Then we proved that such algorithms can be transformed into distributed algorithms without loss it asymptotic time complexity and constant-size labels. Finally, we prove that the optimal broadcasting algorithm from \cite{DBLP:journals/jcss/KowalskiP13} is monotone domination resistant which gives a constructive algorithm constant size labels and optimal time.
}
Table~\ref{tab:broadcast} compares algorithms from \cite{DBLP:conf/spaa/EllenGMP19,EllenG20} which inspired our result, to our algorithms that improve them, presented in this section.

\begin{table}[h]
\begin{center}
\begin{tabular}{|c|c|c|c|c|}
\hline
Algorithm & Section/Theorem & Constructive &  Time & Paper \\
\hline
\hline
$\mathcal{B}$ & -- & Yes & $O(n)$ & \cite{DBLP:conf/spaa/EllenGMP19} \\
\executora & -- & No & $O(D\log n +\log^2n)$ & \cite{EllenG20}\\ 
\executora & -- & Yes & $O(D\log^2 n)$ & \cite{EllenG20}\\
\hline
\Levelled \executora & S.~\ref{sec:lev:alg}/Th.~\ref{th:alg:lev:exec:non} & No & $O(D\log n +\log^2n)$ & here \\
\Levelled \executora & S.~\ref{sec:alg:lev:constr}/Cor.~\ref{broadcastCorollary} & Yes & $O(D\log^2 n)$ & here \\
\fasta & S.~\ref{sec:alg:broadcast:our:constr}/Th.~\ref{th:fast:alg} & Yes & $O(D+\min(D,\log n)\cdot$ & here \\
& & & {$\log^2 n)$} & \\
\fastesta & S.~\ref{sec:fastest:alg}/Th.~\ref{th:fastest:alg} & No & $O(D+\log^2 n)$ & here \\
\hline
\end{tabular}
\end{center}
\caption{Overview of broadcasting algorithms in Section~\ref{sec:broadcasting}.}
\label{tab:broadcast}
\end{table}

\subsection{\Levelled \executoraNOS}\label{sec:lev:alg}
%\input{../broadcast/layered_executor}

%\tj{to moze gdzies indziej, wczesniej?\\}
We will say that a broadcasting algorithm is a \emph{levelled algorithm}, if each node on level $l>0$ {\em accepts} the broadcast message only if it is received from a node on level $l-1$. In other words, nodes from level $l$ accept the broadcast message from nodes on level $l-1$ and they ignore the broadcast message if it is received from a node on level $l$ or $l+1$. Observe that, in order to incorporate such requirement, it is necessary to provide to nodes information which permits them to deduce the level of a node from which they receive the broadcast message.

%For completeness and in order to facilitate exposition of our final algorithm, we 
We present the \Levelled \executora below and state its properties.
%
%In this subsection we will describe Levelled \executora from \cite{EllenG20} in two ways. 
First, we modify the \executora so that it becomes a levelled algorithm, while preserving time complexity and various other properties of the \executoraNOS. 
%It is an algorithm based on dominating set idea. 
%%
%\tj{TO ZDANIE GDZIE INDZIEJ?: \\}
%{\color{red}Secondly, we replace Chlamtac-Weinstein combinatorial result (Lemma~\ref{?}) result assuring constructive 
%variant of the labeling scheme for \executor with a more efficient construction based on \cite{DBLP:journals/dc/KowalskiP07}. ?}

The \executora combines the \emph{dominating set mechanism} introduced in \cite{DBLP:conf/spaa/EllenGMP19} with the \emph{propagation mechanism} from \cite{EllenG20}.
An execution of the algorithm is divided into blocks, each consisting of 3 steps: \textit{Broadcast} step, \textit{Feedback} step and \textit{Go} step. 
Let $r$ be the number of the current %round, 
block and let
%
%then 
$\mathbf{INF}_r$ and $\mathbf{UNINF}_r$ be respectively the sets of informed nodes and uniformed nodes %in %round 
at the beginning of the block
$r$. In 
the
%each 
%round 
block
$r$ we try to keep some minimal set  $\mathbf{DOM}_r$ of informed nodes which dominates frontier nodes, i.e., uninformed neighbors of informed nodes.
Formally, $\DOM_r$ is a minimal subset of $\INF_r$ with respect to set inclusion, such that, for each node $v \in \UNINF_r$, if there exists a neighbor  $u \in \INF_r$ of $v$, then there is some neighbor $u'$ of $v$ in the dominating set $\DOM_r$.

%$\forall v \in \mathbf{UNINF}_r \left(\exists u \in \mathbf{INF}_r.\{u,v\} \in E \right) \implies \left(\exists u' \in \mathbf{DOM}_r.\{u',v\} \in E \right) $. 

The labels of nodes for the \executora are assigned gradually during a simulation of the algorithm on a given input graph. As it might be seen as a simulation of a randomized distributed algorithm based on ideas from \cite{Bar-YehudaGI92}, we assign the values of some bits of the labels randomly.
However, as the simulated distributed algorithm accomplishes broadcast in the claimed number of rounds with high probability, one can assign those bits by simulating all possible random choices of the randomized algorithm and fixing the choices assuring the given round complexity.
At the beginning we set $\DOM_0 = \{s\}$, where $s$ is the source node for the given instance of the broadcast problem. In the \textit{Broadcast} step of the block $r$, all nodes in $\DOM_r$ transmit theirs message. 
We preserve the invariant that each node is always aware if it belongs to $\DOM_r$ in the current block $r$.
As $\DOM_r$ is minimal, there will be at least one node newly informed only by $v$ for each node $v$ from $\DOM_r$. One of the nodes informed by $v$ in block $r$ is chosen to be the \textit{Feedback Node} of $v$ in that block.  
\begin{comment}
If $u$ is the Feedback Node of $v$ in some block, we also say that $v$ is the anti-feedback node of $u$ and denote it as $\text{antifeed}(u)=v$. Note that each node $u$ might be an anti-feedback for at most one other node $v$, since the only ``candidate'' for its anti-feedback node is the only node which successfully transmitted the broadcast message to $u$ in some block, while $u$ has not received the broadcast message in any earlier block.
\end{comment}

All newly informed \textit{Feedback Nodes}	send their values of \textit{Stay} and  \textit{Go} bits in the \textit{Feedback} step of the block, provided that at least one of them (\textit{Stay} or \textit{Go}) is equal to $1$.\footnote{It is essential that a node $v$ with $\mathit{Stay}=$ $\mathit{Go}=0$ remain silent in the \textit{Feedback} step of the block $r$ in which $v$ becomes informed. In this way one can assure that at most one feedback node of $w$ in $\DOM_r$ sends a message to $w$ collision-free in the \textit{Feedback} step.} If the \textit{Feedback Node} of $v$ sends the value $\mathit{Go}=1$ then $v$ will also broadcast in the ``bonus'' \textit{Go} step of the block. The value of \textit{Go} bit is chosen at random by the oracle for all \textit{Feedback Nodes}. Finally, $v$ will stay in the dominating set $\DOM_{r+1}$ if and only if the $\mathit{Stay}$ bit of its feedback node is equal to 1.\footnote{Recall that, the bits of labels ``instruct'' nodes how to execute the algorithm in the distributed setting. Thus, no node $v$ has the centralized view of the progress of the algorithm allowing it to determine whether in belongs to the set $\DOM_r$ in the current block $r$. Therefore, we provide this information to nodes through labels.} Each newly informed node joins the dominating set  $\DOM_{r+1}$ if its \textit{Join} bit is equal to $1$.

The \Levelled \executora is a slight modification of the \executora described above. %To achieve that 
In order to transform the \executora into a levelled algorithm,
we  associate a new variable $\Lev(v)$ with each node $v$ 
such that $\Lev(v)$ is equal the value of the distance from the source node $s$ to $v$ modulo $3$. (That is, $\Lev(v)=\level(v)\mod 3$.)
The value of $\Lev(v)$ can be stored as two extra bits of the label of $v$.
Assume that the current block is $r$ such that $r \mod 3 = x$.
%In %steps 
%blocks $r$ such that  
Then, only nodes from the current dominating set with $\mathit{Level}$ value equal to $x$ transmit in \textit{Broadcast} and \textit{Go} Steps of the block $r$, and listen in the \textit{Feedback} step.
Moreover, only nodes with $\mathit{Level}$ value equal to $(x+1) \mod 3$ listen in the \textit{Broadcast}, \textit{Go} steps and transmit in the \textit{Feedback} step. 
In order to adjust this change to the standard \executora preserving the general meaning of the sets $\DOM_r$, $\INF_r$ and $\UNINF_r$ as described above,
%that 
we must also redefine the sets $\DOM_r$, $\INF_r$ and $\UNINF_r$ so that, for $r\mod 3=x$, one is looking only at edges where informed nodes are at such levels $l$ that $l\mod 3=x$ and, for an informed node at level $l$, only its uninformed neighbors at level $l+1$ are considered. 

The final label for each node $v$ is the tuple consisting of four elements: 
\begin{itemize}
\item \textit{Join} is the bit indicating whether $v$ should join the dominating set when it receives the broadcast message for the first time from a node at the level preceding the level of $v$ and remain in the dominating set for the next $3$ blocks (recall that nodes are active as informed ones only in one of each three blocks in the \Levelled \executoraNOS).

\item $\Lev$ is the distance from the root to $v$ modulo $3$.
\end{itemize}
Assume that $v$ is informed in the block $r$ such
that $\Lev(v)=(r+1)\mod 3$. If $v$ is chosen as the feedback node of some node $w$ which informed $v$ then
\begin{itemize}
\item \textit{Stay} is the bit indicating whether $w$ should stay in the dominating set $\DOM_{r+3}$.
\item \textit{Go} is the bit indicating whether $w$ should broadcast in the \textit{Go}$_{r+3}$ step of the block $r$.
\end{itemize}
If $v$ is not a feedback node of any other node then the bits \textit{Stay} and \textit{Go} are set to $0$.
Finally, $v$ sends the bits \textit{Stay} and \textit{Go} in the \textit{Feedback} step of block $r$ if \textit{Stay}=1 or \textit{Go}=1.

At the beginning we assign the label $(1,0,0,0)$ to the source node $s$. 
The labels of all other nodes are initiated to zeroes and they will be assigned gradually during a simulation of the final
distributed algorithm using those labels, as described above.
%All other nodes will be given while simulating the algorithm, as described above.

By a slight modification of proofs from \cite{EllenG20} (see Theorem~12 in \cite{EllenG20}), one can show the following non-constructive result.
%For the details of this algorithm one can refer to \cite{EllenG20}.
%\tj{(GDZIES WSPOMNIEC O TYM, ZE BITY Go losowo)}

\begin{theorem}\label{th:alg:lev:exec:non}\cite{EllenG20}
There exists %such 
a labeling scheme of constant length for which the \Levelled \executora finishes broadcast in $O(D\log n+\log^2n)$ rounds. 
\end{theorem}

\subsubsection{Constructive Variant of \Levelled \executora} \label{sec:alg:lev:constr}

As in \cite{EllenG20}\ (Lemma~13), we will use the following lemma which is a slight modification of the result proved by Chlamtac and Weinstein \cite{DBLP:journals/tcom/Chlamtac91}.

\begin{lemma}\label{DetBroadcastLemmaShort}\cite{DBLP:journals/tcom/Chlamtac91,EllenG20}
	Let $G$ be a bipartite graph with bipartition sets $A$ and $B$, where the degree of each node $v\in B$ is at least one. 
 Then, there exists a polynomial time %an efficient 
 deterministic algorithm which finds the sets $A'\subseteq A$ and $B' \subseteq B$ such that $|B'| \geq |B|/\left(15\log |A|\right)$ and each node from $B'$ has exactly one neighbor in $A'$.
\end{lemma}
The above result is obtained as follows. One can show that a random choice of $A'$ 
guarantees that the expected size of $B'$ is
at least $|B|/\left(15\log |A|\right)$.
%The random choice of $A'$ goes as follows for some 
Let $A'$ be chosen randomly such that,
for some fixed
$p\in\{1/2^1, 1/2^2,\ldots,1/2^{\ceil{\log n}}\}$, $v$ is chosen to belong to $A'$ with probability $p$, independently of random choices for other elements of $A$, for each $v\in A$.
Then, using the technique of derandomization  with maximization of expectation one can determine $A'$ satisfying Lemma~\ref{DetBroadcastLemmaShort} in polynomial time. For more details, refer e.g.,\ to \cite{EllenG20}.

%Here
We now describe a deterministic assignment of the bits \textit{Go}, using Lemma~\ref{DetBroadcastLemmaShort}.
Divide an execution of the algorithm into consecutive \emph{stages} 
consisting of $T = 15 \log^2 n$ of our $3-step$ blocks. If a node $v$ gets informed 
%in round 
during the block $t_v$ then \textit{Go} bits in its feedback nodes in the %round 
blocks $t'>t_v$ will be set to $0$ until the block $t'_v$ equal to the smallest number $t'$
such that 
$t'=0 \mod T$. (Note that it is sufficient that $\log n$ is known to the oracle assigning labels
and the distributed broadcast algorithm working at nodes using these labels can work without knowledge of $\log n$.)
%which is the smallest 
%round 
%block
%
%such that $t'_v > t_v$ and $t'_v = 0 \mod T$. 
In other words, \textit{Go} bits of $v$ are set to $0$ until the end of the stage in which
$v$ is informed. 

Now, we are ready to describe an efficient algorithm which assigns the bits \textit{Go} to all nodes, together with other bits of labels described above.
Let $k \geq 0$ be an integer, let $A$ be the subset of the set of nodes $\DOM_{kT}$ located on some level $l$ and let all labels be already fixed for nodes informed before the block $kT$.
%from the dominating set in %round 
%block
%
%$kT$ 
That is, for a fixed $l$, $A=\DOM_{kT}\cap L_l$.
% such that $v\in \DOM_{kT}$ belongs to $A$ iff the level of $v$ is equal to $l$.
%
Moreover, for $r\ge 0$, let $A_r \subseteq A$ $\subseteq L_l$ be the set of nodes from  $A$ still being in the dominating set of
%round 
the block
$kT + r$. % on the level $l$. 
That is, $A_r$ is the subset of $A$ containing the nodes which are in $\DOM_{kT+r}$.
Let $B_r$ be the set of uninformed neighbors of $A_r$ on the level $l+1$. 
Then, let $A'_r\subseteq A_r$ and $B'_r\subseteq B_r$ be the sets provided by the the algorithm from Lemma~\ref{DetBroadcastLemmaShort} applied to the bipartite graph induced by the sets $A_r$ and $B_r$. 
%
%Given these sets,  for 
For each $v \in A_r$, the bit \textit{Go} of the feedback node of $v$  in block $kT+r$ is set to $1$ iff $v \in A'_r$.
%and it is set to $0$ otherwise.

\begin{lemma} \label{l:aux:bipartite}
	The set $B_{r'}$, for $r'=15\log^2 n $, is empty.
 %, where $B_{r'}$ is defined \tj{as in the above paragraph.}
\end{lemma} 
\begin{proof}
	Fix any 
	%round 
	block
	$r < r'$. By the construction of the algorithm, any node $v \in A_r$ in the $(kT + r)$th 
	%round 
	block
	gets the \textit{Go} bit from its feedback node. 
    If $v \in A'_r$ then its bit \textit{Go} is set to $1$ and therefore it broadcasts in the \textit{Go} round. 
    By Lemma \ref{DetBroadcastLemmaShort}, if all nodes from $A'_r$ transmit simultaneously then all nodes from $B'_r$ receive the broadcast message and therefore they are informed after the block $kT+r$. As $B_{r+1} \subseteq B_r \setminus B'_r$, Lemma~\ref{DetBroadcastLemmaShort} implies that the size of $B_{r'}$ can be bounded as follows:
		$$ |B_{r'}| \leq |B_0| \left(1-\frac{1}{15\log n}\right)^{15\log^2 n} \leq |B| \exp(-\log_2 n) =\frac{|B|}{n} < 1 .$$
\end{proof}
After becoming informed in some block $t$, a node $v$ waits for the block $r$ with the smallest number $r\ge t$ equal to the multiple of $15\log^2 n$ and then $v$  becomes a member of the ``source part'' of such a bipartite graph containing all its uninformed neighbors.  Broadcasting in a bipartite graph can be done done in $15\log^2 n$ blocks by Lemma~\ref{l:aux:bipartite}. This implies that,
if node $v$ gets first informed in 
	%round 
	block
	$r$ then all %of 
    its 
    %children 
    neighbors are informed by the end of
	%round 
	block
	$r+30\log^2 n$.
%round 
%blocks
%
Hence we get the following corollary.
\begin{corollary}\label{broadcastCorollary}
The \Levelled \executora accomplishes broadcast in $O(D\log^2n)$ rounds, and uses a constructive labeling scheme.
\end{corollary}

\commentt{
\subsubsection{Ellen\&Gilbert Algorithm vs. its Levelled Variant}
Ellen and Gilbert \cite{EllenG20} have not introduced nor analyzed the levelled variant of \executoraNOS. However, we can argue that their analysis of correctness and efficiency of \executora applies to \Levelled \executora as well. Observe that the only difference between the levelled variant of the algorithm and the standard variant is that, for each level $l\ge 0$ and each node $v$ with $\level(v)=l$, $v$ ignores messages received from the nodes from the levels $l$ and $l+1$ during an execution of the \levelled \executora while these messages are not ignored in the standard \executoraNOS. However, in the analysis of \executora, it is sufficient to analyze progress of the broadcast message only on shortest paths from the source to all other nodes. This however means in particular that, in order to determine the fact of reception of the broadcast message by a node $v$ at the level $l$ and round number of this reception, one may ignore the fact that the broadcast message for was delivered earlier than to $v$ for some other nodes from levels $l'\geq l$.

%First, given an instance of the brodcast problem, let us conceptually remove the edges between nodes at the same level of a BFS tree rooted at the source vertex $v$. For such an input graph the analysis from \cite{?} works
%.... UPS .... MOZNA DOSTAC WIADOMOSC OD WIERZCHOLKOW Z NASTEPNEGO LEVELA... CHYBA
%czyli trzeba troche ``porzezbic''....
}

%\subsection{Broadcast Algorithm with Linear Dependency on $D$}\label{sec:alg:broadcast:our:constr}
\subsection{\fastaNOS: Linear Dependency on $D$}\label{sec:alg:broadcast:our:constr}

%\textbf{Trzeba zmienic definicje x(v) i usunac wszysztkie wstawki z ppb i zmienic stale}

In this section, we extend the labeling scheme and modify the \Levelled \executora in order to 
improve its time complexity {to $O(D+\min(D,\log n)\log^2n)$.} The first summand in this complexity (the linear dependency on $D$) is cleary optimal, and the second summand will be improved later. 
%remove the multiplicative factor $\Theta(\log n)$ in its round complexity. 
%
The resulting algorithm is called {\fastaNOS}.

As before, we build labels gradually starting from zero labels, by simulating the final broadcasting algorithm step by step. 
The final algorithm will work in the framework of the  {\Levelled} \executoraNOS. 
Significantly, we introduce special \emph{shortcut edges}, as described by Gasieniec et al.\ \cite{DBLP:journals/dc/GasieniecPX07}. As further explained, all but $O(\log n)$ edges are such shortcuts on each path from the source of a BFS tree which is a \thrtNOS. Interestingly, these shortcut edges correspond to fast edges of a 2-height respecting BFS spanning tree of the communication graph.

One can refer to % Chrobak et al.\ 
\cite{DBLP:journals/dc/GasieniecPX07,DBLP:journals/iandc/ChrobakCGK18} for the proof of the following lemma. 
%\tj{DLACZEGO NIE DO \cite{DBLP:journals/dc/GasieniecPX07}?}
\begin{lemma}\cite{DBLP:journals/iandc/ChrobakCGK18}\label{l:2h:max}
	For a tree of size $|V|=n$ the maximum value of 2-height is $\log n$.
\end{lemma}
{As $h_2(\parent(v)\leq h_2(v)+1$, for each node $v$ which is not the root of the considered tree, we see that the maximum value of $2$-height for a given spanning tree $T$ is not larger than the height of $T$.
Moreover, as the height of a BFS tree of $G=(V,E)$ is $O(D)$, we have the following corollary.
\begin{corollary}\label{cor:h2}
    For a BFS spanning tree of a graph $G=(V,E)$, the maximum value of $2$-height is at most $\min(D,\log n)$.
\end{corollary}
}

%A NIE MOZEMY PO PROSTU NAJPIERW ZBUDOWAD 2-HEIGHT-RESPECTING-TREE? 
Similarly as the \Levelled \executoraNOS, the {\fastaNOS}
works in blocks which consist of a fixed constant number of steps. It uses the labels for the \Levelled \executora and performs its standard steps \textit{Broadcast}, \textit{Feedback} and \textit{Go} in each block.
Moreover, two additional steps \textit{Fast} and \textit{Rescue}, as well as some additional bits of the labels are added in order to accelerate the broadcasting process, partially using some ideas of the centralized broadcasting algorithm from \cite{DBLP:journals/dc/GasieniecPX07}.
%During the simulation of the algorithm the oracle will apply the preprocessing algorithm making our arbitrary BFS Tree a \thrt $T$? 
To this aim, in the process of assignment of labels combined with the simulation of the distributed algorithm, we start by building a  BFS tree $T$ of the input graph with the source vertex $s$ as the root of $T$, which is a \thrtNOS. 
One can build such a tree for each graph, as argued in Lemma~\ref{lem:2hrt}.
(Let us stress here that a 2-height respecting tree must satisfy stricter requirements than the so-called gathering spanning trees from \cite{DBLP:journals/dc/GasieniecPX07}, and this difference is essential for our labeling scheme and for the distributed algorithm.)
Then, we will follow the framework of the \Levelled \executora extended in such a way that 
\begin{itemize}
\item
Each block is extended by rounds \textit{Fast} and \textit{Rescue} devoted to acceleration of the broadcasting process through non-colliding shortcut/fast edges connecting nodes with their children, so that the $2$-height of a node and its child are equal.
\item
In order to instruct nodes about their actions in these new steps in blocks, we also add two extra bits to their labels, called the \textit{Fast} bit and the \textit{Rescue} bit. As it turns out, facilitating these transmissions without collisions faces significant challenges which we %discuss together with solutions in the following part of this section.
overcome by taking advantage of the fact that $T$ is a \gt.
\end{itemize}
%
%This will let the nodes on the same level with the same \textit{2-height} to transmit a message to their children with the same \textit{2-heights} without any collisions using our \textit{fast edges}. 
\commentt{
    As the value of \textit{2-height} can change at most $\log n$ times on each path from the source to a leaf of a \thrt, we get that broadcasting will be done in time $O(D + \log^3 n)$. Below, we provide our construction and the proof of its correctness in more detail.
}

%

%More specifically 

%we want 
First,
for each node $v$ we define the 
%\textit{ultimate round} 
\textit{ultimate block number} 
$x(v)$ of $v$ as follows
$$
x(v) = 3\cdot \left(\level(v) + 30 \ceil{\log^2 n} (h_2(r) - h_2(v))\right) + \left((\level(v)-1) \mod 3 \right).
$$
We will now build the rest of the algorithm to ensure that $x(v)$ is an upper bound on the number of %round
the block
when $v$ is informed, i.e., when it receives the broadcast message for the first time, and this received message is transmitted from a node on the level $\level(v)-1$.

%
%In order to achieve the above goal, we will add two new bits: \textit{Fast} and \textit{Rescue} to the labeling of {\Levelled} \executoraNOS. 
%
%Moreover, we  add two new steps to each block of {\Levelled} \executoraNOS: \textit{Fast} step and \textit{Rescue} step. 
To achieve the above goal, we set the bits \textit{Fast} and \textit{Rescue} so that their values instruct nodes about their actions in the \textit{Fast} and \textit{Rescue steps} respectively, facilitating efficient usage of fast edges.
Initially all bits \textit{Fast} and \textit{Rescue} for all nodes are set to $0$ which means that (initially) no nodes are supposed to transmit in the new steps \textit{Fast} and \textit{Rescue} of any block. Additionally, for each  $u$, $v$ such that $u$ is the feedback node of $v$ in some block, apart from its \textit{Stay} and \textit{Go} bits, $u$ will also transmit its values of \textit{Fast} and \textit{Rescue} bits to $v$ in the appropriate Feedback step.

%Assignment of the values of new bits \textit{Fast}, \textit{Rescue} are intended to manage their
%actions in the Fast and Rescue steps of a block $r$:
%\begin{itemize}
%\item
%A newly informed node $v$ transmits in \textit{Fast} (\textit{Rescue}, respectively) step of the current block if
%he value of its $\textit{Fast}_J$ bit ($\textit{Rescue}_J$, respectively) is equal to $1$.
%\item
A node $v\in \DOM_r$ such that $\Lev(v)=r\mod 3$ (i.e., informed in a block $r'<r$) transmits in \textit{Fast} (respectively \textit{Rescue}) step of the current block if
the value of received $\textit{Fast}$ (respectively $\textit{Rescue}$) bit transmitted to $v$ (in the \textit{Feedback} step) by its current feedback node is equal to $1$.
%\end{itemize}

Below, we describe assignments of the bits \textit{Fast} and \textit{Rescue}
%, their impact on behaviour of nodes and their properties 
in more detail.

%\tj{TU USUNALEM AKAPIT O TYM JAK WIERZCHOLEK $p(v)$ POLICZYMY SOBIE $x(v)$!!}
%During an execution of the algorithm we can encounter the three cases analyzed below. 
%In all of them we will look at the 
%\begin{comment}
   
If $h_2(v)\neq h_2(\parent(v))$ or $\parent(v)$ has not received the broadcast message until the round $x(v)>x(p(v))$ (see the definition of $x(v)$), no change of the values of the bits \textit{Fast} or \textit{Rescue} are caused by $v$.
Similarly, if the node $v$ receives the broadcast message before the block $x(v)$ then it does not cause change of bits \textit{Fast}, \textit{Rescue} of any node and therefore does not cause additional transmissions in the steps \textit{Fast} or  \textit{Rescue} of any block.

So let $v$ be a  node 
such that $h_2(v) = h_2(\parent(v))$, i.e., it is connected to its parent by a fast edge.
Additionally assume that the node $v$ has not received the broadcast message from $\parent(v)$
until the beginning of the block $x(v)$ but 
$\parent(v)$ received the broadcast message in the block $x(\parent(v))$ or earlier. 
\commentt{
Note that, under these assumptions,
\begin{equation}\label{eq:mod}
	\begin{tabular}{rcl}
		$x(v) - x(\parent(v))$ &=& 
		$3 + ((\level(\parent(v)) + 1) \mod 3) - (\level(\parent(v)) \mod 3)$ \\
		&=& $\left\{
		\begin{tabular}{cll}
			4 &  \mbox{ when } & $\level(\parent(v))\mod 3\in\{0,1\}$, \\
			1 &  \mbox{ when } & $\level(\parent(v))\mod 3=2$. \\
		\end{tabular}\right.$
	\end{tabular}
\end{equation}
is the value that each node with level equal to $\level(\parent(v))$ can compute.
%\end{comment}
}

%The cases are the following:
%Recall that, according to the inductive hypothesis, the node $p(v)$ was informed until the round $x(\parent(v))<x(v)$
%

%Now, assume that $v$ has not received the broadcast message until the block $x(v)$.
Consider the following cases regarding the status of node $p(v)$ at the beginning of block $x(v)$:
\begin{enumerate}
\item[\textbf{Case~1.}] 
    The node $\parent(v)$ was informed before the block
    $x(v)$ and it is in the dominating set $\DOM_{x(v)}$ of 
    the block $x(v)$ of the algorithm.

    Then $\parent(v)$ has some feedback node $w$ first informed in 
    the block
    $x(v)$. We set $Fast$ bit of $w$ to $1$. 

\item[\textbf{Case~2.}]  
    The node $\parent(v)$ was informed before the block
    $x(v)$ but it is not in the dominating set $\DOM_{x(v)}$ of 
    the block $x(v)$ of the algorithm.

 As $v$ is still not informed before the block $x(v)$, there must be a neighbor $u$ of $v$ in the dominating set $\DOM_{x(v)}$ with $\level(u) = \level(\parent(v))$. 
 Then, as $u$ is in $\DOM_{x(v)}$, it has some feedback node $w$ first informed, by $u$, in 
 the block $x(v)$. We set the \textit{Rescue} bit of $w$ to $1$. 
\end{enumerate}

The %description 
 above  cases complete the %description of the 
description of \fastaNOS.
%\tj{POTRZEBUJEMY DOBREJ NAZWY DLA NASZEGO ALGORYTMU}
In Table~\ref{tab:local} we summarize the broadcast algorithm from the local perspective of nodes which initially only know their labels. We also express properties of the algorithm following from the assignment of the bits \textit{Fast} and \textit{Rescue} and the actions of nodes caused by these setting in the following observation.
\begin{comment}
\begin{table}[h]
\begin{center}
\begin{tabular}{|c|c|c|}
  \hline \hline 
  Step & Newly informed node $v$ & $v\in \DOM_r$, informed earlier\\
  \hline
  \hline
  \textit{Broadcast} & silent & sends $B$\\
  \hline
  \textit{Feedback} & %If $\textit{Stay}=1$, 
  sends 
  & silent \\
  (regards block $r-1$) & the bits \textit{Stay, Go, Fast$_F$, Rescue$_F$ }& \\
  \hline
  \textit{Go} & silent & sends $B$ if received \textit{Go}$=1$ \\
  & & in \textit{Feedback} step\\
  \hline
  \textit{Fast} & sends $B$ if \textit{Fast}$_J=1$ & sends $B$ if received \textit{Fast}$_F=1$ \\ & & in \textit{Feedback} step\\
  \hline
  \textit{Rescue} & sends $B$ if \textit{Rescue}$_J=1$ & sends $B$ if received \textit{Rescue}$_F=1$ \\
  & & in \textit{Feedback} step\\
  \hline
  \hline
\end{tabular}
\end{center}
\caption{Behaviour of a node in the steps of a block $r$, where $B$ is the broadcast message and $r\mod 3=\Lev(v)$. Here, a \emph{newly informed} node is a node which was informed in the block $r-1$ and a node \emph{informed earlier} was informed in a block $r'<r-1$.}
\label{tab:local}
\end{table}    
\end{comment}
\begin{table}[h]
\begin{center}
\begin{tabular}{|c|c|c|}
  \hline \hline 
  Step & Newly informed node $u$ such & $v\in \DOM_r$ informed earlier\\
  & that $r\mod 3=(\Lev(u)-1)\mod 3$ & such that $\Lev(v)\mod 3=r\mod 3$\\
  \hline
  \hline
  \textit{Broadcast} & silent & sends $B$\\
  \hline
  \textit{Feedback} & %If $\textit{Stay}=1$, 
  if $Stay \neq 0$ or $Go \neq 0$ then $u$ sends 
  & silent \\
   & the bits \textit{Stay, Go, Fast, Rescue }& \\
  \hline
  \textit{Go} & silent & sends $B$ if received \textit{Go}$=1$ \\
  & & in \textit{Feedback} step\\
  \hline
  \textit{Fast} & silent & sends $B$ if received \textit{Fast}$=1$ \\ & & in \textit{Feedback} step\\
  \hline
  \textit{Rescue} & silent & sends $B$ if received \textit{Rescue}$=1$ \\
  & & in \textit{Feedback} step\\
  \hline
  \hline
\end{tabular}
\end{center}
\caption{Behaviour of nodes in the steps of a block $r$, where $B$ is the broadcast message.
%and $r\mod 3=\Lev(v)$. 
%Here, a \emph{newly informed} node is a node which was informed in the block $r-1$ and a node \emph{informed earlier} was informed in a block $r'<r-1$.
}
\label{tab:local}
\end{table}

\commentt{
**** OLD BEGING
\begin{observation}\label{obs:fast:rescue}
\begin{enumerate}
    \item 
    If a node $v$ transmits in the Fast step of a block $b$ then $b=x(v')$ for a child $v'$ of $v$ such that $h_2(v')=h_2(v)$, i.e., $(v,v')$ is a fast edge.

    \item
    If a node $v$ transmits in the Rescue step of a block $b$ then there exists a node $v'=\phi(v)$ 
    such that
    \begin{enumerate}
        \item $\level(v')=\level(v)+1$,
        \item $(v,v') \in E$, 
        \item $b=x(v')$ and,
        \item $v'$ is connected with its parent by a fast edge. 
    \end{enumerate} 
    Moreover, for each node $v'$ there exists at most one $v$ such that $v'=\phi(v)$ satisfying the properties (a)--(d).
    %Moreover, the assignment $\phi$ is one-to-one: if $\phi(u)$ and $\phi(w)$ are defined for $u \neq w$ then $\phi(u)\neq \phi(w)$.
\end{enumerate}
\end{observation}
*** OLD END
}
\begin{observation}\label{obs:fast:rescue}
\begin{enumerate}
    \item 
    A node $v$ transmits in the Fast step of a block $b$ if and only if $b=x(v')$ for a child $v'$ of $v$ such that $h_2(v')=h_2(v)$, i.e., $(v,v')$ is a fast edge and $v'$ is not informed before the block $b$.

    \item
    Let $V_b$ be the set of nodes which transmit in the \textit{Rescue} step of a block $b$ and
    let $V'_b$ be the set of such nodes $v'$ which are uninformed before the block $b$, $x(v')=b$ and their parents are not in the dominating set $\DOM_b$.
    Then, there exists a one-to-one assignment $\phi_b: V_b\to V'_b$ which satisfies the following property.\\
    Let $v'=\phi _b(v)$ for $v\in V_b$. Then:
    %such that
    \begin{enumerate}
        \item $\level(v')=\level(v)+1$,
        \item $(v,v') \in E$, 
        %\item $b=x(v')$ and,
        %\item $v$ is not the parent of $v'$ in $T$,
        \item $v'$ is connected with its parent $p(v')$ by a fast edge, i.e., $h_2(v')=h_2(\parent(v'))$. 
    \end{enumerate} 
    %Moreover, for each node $v'$ there exists at most one $v$ such that $v'=\phi(v)$ satisfying the properties (a)--(d).
    %Moreover, the assignment $\phi$ is one-to-one: if $\phi(u)$ and $\phi(w)$ are defined for $u \neq w$ then $\phi(u)\neq \phi(w)$.
\end{enumerate}
\end{observation}

%As it is a  
As the above described  \fasta is an extension of the
{\Levelled} \executora 
%with 
by adding 
extra transmissions in each block, 
the correctness of the new algorithm follows directly from the correctness of the  {\Levelled} \executora
assured by the domination mechanism. 
To obtain a bound on the time of the new algorithm, we  prove the following lemma.
\begin{lemma}\label{l:alg:our:cases}
	A node $v$ becomes informed by the end of block $x(v)$, for each $v\in V$.% with probability at least 
	%	$ 1 - \frac{h_2(r) - h_2(v)}{n^2}$.
\end{lemma}
\begin{proof} \label{p:l:2h}
	%A proof will be done by induction on value of $x(v)$. 
    We prove the lemma by induction on $x(v)$.
    For the base case $x(v)=0$,
    it is sufficient to consider only the source vertex $s$ as $x(s)=0$ and $x(v)>0$ for each $v\neq s$. As $s$ is informed at the beginning of an execution of the algorithm, the base case is 
    %trivially proven.
    satisfied.
	
    For the inductive step, consider an arbitrary value $x>0$ and any node $v$ such that $x(v) = x$. Assume that the lemma %stands 
    holds
    for all nodes $u$ such that $x(u) < x$. 
    As $x(v)>0$, the node $v$ is not the source node $s$. 
    So the parent $\parent(v)$ is defined in this case. 
    As $\level(v) = \level(\parent(v)) + 1$ and $h_2(v) \leq h_2(\parent(v))$, we know that $x(\parent(v)) < x(v) =x$. Therefore, by the inductive hypothesis, $\parent(v)$ was informed until the 
    %round 
    block
    $x(\parent(v))<x=x(v)$.
	Consider the following cases:
    \begin{itemize}
    \item	$h_2(v) < h_2(\parent(v))$
		
		%Let us estimate the probability that $v$ does not get the message by round $x(v)$ in that case. 
		As \text{2-heights} of $v$ and $\parent(v)$ are different, the values of $x(v)$ and $x(\parent(v))$ differ by more than $90 \ceil{\log^2n}$ which means that 
        $p(v)$ is active in at least $30\log^2n$ blocks
        of the \levelled algorithm
        between the blocks $x(v)$ and $x(p(v))$. 
        %two rounds. 
        Then, by Corollary~\ref{broadcastCorollary}, as $\parent(v)$ gets informed 
        until the block $x(\parent(v))$, in view of the inductive hypothesis, $v$ gets informed by the block 
        $x(\parent(v)) +90 \ceil{\log^2n} = x(v)$.

    \item $h_2(v) = h_2(\parent(v))$.
		
        If  the node $v$ is informed before the block $x(v)$, we are done. If not, we will inspect carefully Cases~1--2 which might appear during an execution of our algorithm, presented above in the description of the algorithm.
		
	For Case~1,
        observe that the node $\parent(v)$ sends the broadcast message in the \textit{Fast} step of 
        the block $x(v)$ -- see Observation~\ref{obs:fast:rescue}.1. 
        Thus, it is sufficient to show that $v$ receives this message transmitted by $\parent(v)$ in the \textit{Fast} step of the block $x(v)$. To get a contradiction, assume that there is another node $w$  transmitting in the \textit{Fast} step of the block
        $x(v)$, such that $w\neq \parent(v)$ is a neighbor of $v$ and therefore its transmission causes collision at $v$ in the \textit{Fast} step of the block $x(v)$ -- see Figure~\ref{fig:alg:analysis:cases}. 
        By Observation~\ref{obs:fast:rescue}.1, there exists $w'$ connected to $w$ such that 
        \begin{itemize}
            \item $w'$ is a child of $w$ connected to $w$ by a fast edge, and thus $h_2(w)=h_2(w')$;
            \item $x(v)=x(w')$ (since $w=\parent(w')$ transmits in the block $x(w')$ which is the same as the transmission block of $p(v)$ equal to $x(v)$) and, as $\level(v)=\level(w')$, the equality $h_2(v)=h_2(w')$ also holds.
        \end{itemize}
        The above relationships imply not only that $\level(\parent(v)) = \level(w)$ but also $h_2(\parent(v)) = h_2(w)$, since they are connected by fast edges with their children $v$, $w'$ such that $h_2(v)=h_2(w')$.
        %; c.f.\ Fig.~\ref{fig:alg:analysis:cases}.
        \begin{figure}
        \centering
        \includegraphics[width=0.75\linewidth]{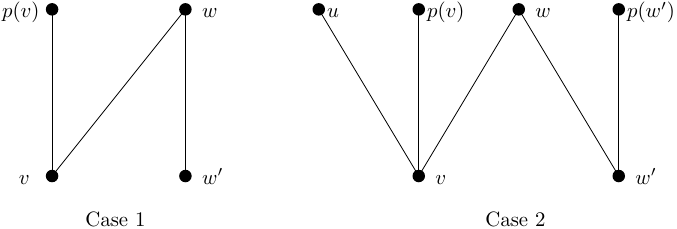}
        \caption{The cases in the analysis of our algorithm -- Lemma~\ref{l:alg:our:cases}.}
        \label{fig:alg:analysis:cases}
        \end{figure}        
        Thus both $w'$ and $v$ have the same \textit{2-height} as their respective parents. 
        This final setting (see Fig.~\ref{fig:alg:analysis:cases} again) implies 
        that existence of the edges $(v,w)$ and $(w',w)$ contradicts the requirement %(\ref{d:broadcast:tree:a}) 
        of the definition of a \thrt (see Definition~\ref{def:2h}).

        We can make a similar reasoning for Case~2.
        Here the node $u$, a neighbor of $v$ such that $\level(u)=\level(\parent(v))$ (but $u\neq \parent(v)$ in this case), sends the message in \textit{Rescue} step of %round 
        the block $x(v)$ by Observation~\ref{obs:fast:rescue}.2. To get a contradiction,  assume that there is another node $w$ transmitting in the \textit{Rescue} step of 
        the block $x(v)$ 
        such that $w$ is also a neighbor of $v$ -- see Fig.~\ref{fig:alg:analysis:cases}.
        This fact implies in turn  that $\level(\parent(v)) = \level(w)$. 
        Moreover, by Observation~\ref{obs:fast:rescue}.2, as $w$ is transmitting in \textit{Rescue} step in the 
        block $x(v)$, %it means that it must be 
        $w$ must be a neighbor of some node $w'$ with $x(w')=x(v)$, $\level(w')=\level(v)$ and thus 
        $h_2(w')=h_2(v)$. Additionally, $h_2(\parent(w'))=h_2(w')=h_2(v)$. 
        But these properties imply that simultaneous
        existence of edges $(v,w)$ and $(w',w)$ violates the requirements of the definition of a \thrt (see Definition~\ref{def:2h}).
        %$w$ is a node directly violating requirement 
        %(\ref{d:broadcast:tree:b}) of a \gt.
		
	\end{itemize}
	%Which concludes 
    To summarize, we have proved that, if $\parent(v)$ gets informed by the end of block          $x(\parent(v))$, then $v$ gets informed by the end of block $x(v)$. 
    This concludes the inductive step of the proof of the lemma.	 
\end{proof}

For each node $v$, the value $x(v)$ is bounded by $3\cdot (D + 30 \ceil{\log^2 n} h_2(r)) + 2 \in 
O(D+\log^3n)$. Since blocks are of constant length, {Lemma~\ref{l:alg:our:cases} 
and Corollary~\ref{cor:h2} imply the following theorem.}

\begin{theorem}\label{th:fast:alg}
	The \fasta accomplishes broadcast %a message to all nodes in graph 
    {in $O(D + \min(D,\log n)\log^2 n)$} rounds, and uses a constructive labeling scheme. % with constant probability.
\end{theorem}

\subsection{\fastesta -- Optimization of the polylog Additive Summand} \label{sec:fastest:alg} %%NAMES TO CHANGE

In this section,  we describe a modification of the \fasta which we call the \fastestaNOS. It runs in time only $O(D + \log^2 n) $, which is the optimal broadcasting time even for radio networks of known topology. 
A constant length labeling scheme for this optimal broadcasting can be assigned by an appropriate randomized 
algorithm such that we obtain labels guaranteing $O(D+\log^2n)$ broadcast with high probability.
Thus, using the probabilistic method, we get a nonconstructive constant length labeling scheme supporting broadcasting in optimal time.

The \fastestaNOS\ is almost the same as the \fasta, with the changed  value of the \textit{ultimate block number} and a random assignment of the values of bits \textit{Go}. The random assignment of the bits \textit{Go} is made similarly as in Theorem~\ref{th:alg:lev:exec:non} (Theorem~12 in \cite{EllenG20}). Note however that the change of the definition of the ultimate block number has an impact on the labeling of nodes. Thus, for a given communication graph, the labels for the \fasta and the labels for the \fastesta are usually different. 
%\tj{?CZY TO PRAWDA CO NAPISALEM W TYM AKAPICIE?}

The \emph{ultimate block number} $z(v)$ of a node $v$ is now defined as follows. $z(s)=1$, for the source node $s$, and 
for each $v\in V\setminus\{s\}$,  $z(v)$ is the minimum $z$ satisfying the following condition:\\
$z\mod 6\ceil{\log n}=y(v) \mod 6 \ceil{\log n}$ and $\parent(v)$ becomes informed before block $z$,
%$z(v)=\min\{ z\,|\, z\mod 6\ceil{\log n}=y(v) \mod 6 \ceil{\log n}$ and $(\parent(v))$ becomes informed by the end of block %$z(\parent(v))$,
where
$y(v) = 6\cdot \left(\level(v) + (h_2(r) - h_2(v))\right) + \left(level(v)-1\right) \mod 3$.
\begin{observation}\label{obs:ultimate:z}
If $z(u)=z(v)$ and $\level(u)=\level(v)$ then $h_2(u)=h_2(v)$.
\end{observation}
\begin{proof}
    Assume that $z(u)=z(v)$ and $\level(u)=\level(v)$.
    Thus, according to the definitions, $y(u) \mod 6\log n=y(v)\mod 6\log n$.
    This in turn implies that $(h_2(r)-h(u))\mod 6\log n=(h_2(r)-h(v))\mod 6\log n$,
    but this relationship is satisfied only when $h_2(u)=h_2(v)$, 
    {since $h_2(w)<\log n$ for each node $w$.}
\end{proof}

The labeling scheme for the \fastesta and the algorithm itself are defined as in the case of the \fastaNOS, with two differences in the labeling scheme:
\begin{itemize}
\item 
the assignment of the \textit{Fast},  and \textit{Rescue} bits of each node $v$ is determined by the values of $z(v)$ instead of $x(v)$.
More precisely, assume that $v$ has not received the broadcast message until the block $z(v)$
and $v$ is connected with its parent by a fast edge, i.e., $h_2(v)=h_2(\parent(v))$.

Consider the following cases regarding the status of the node $p(v)$ at the beginning of block $z(v)$:
\begin{enumerate}
\item[\textbf{Case~1.}] 
    The node $\parent(v)$  is in the dominating set $\DOM_{z(v)}$.
    %of the block $z(v)$ of the algorithm.

    Then $\parent(v)$ has some feedback node $w$ informed in 
    the block
    $z(v)$. We set $Fast$ bit of $w$ to $1$. 

\item[\textbf{Case~2.}]  
    The node $\parent(v)$ is not in the dominating set $\DOM_{z(v)}$.
    %of the block $z(v)$ of the algorithm.

 As $v$ is still not informed before the block $z(v)$, there must be a neighbor $u$ of $v$ in the dominating set $\DOM_{z(v)}$ with $\level(u) = \level(\parent(v))$. 
 Then, as $u$ is in $\DOM_{z(v)}$, it has some feedback node $w$ first informed in 
 the block $z(v)$. We will set \textit{Rescue} bit of $w$ to $1$. 
\end{enumerate}
\item
bits $Go$ are assigned as follows:

Define a feedback node of a block $r$ to be each node $v$ such that $v$ is a feedback node of some
node $v'\in\DOM_r$ in the block $r$.
For each block $r$, choose the value $p_r\in\{1,\ldots,\log n\}$ randomly with uniform distribution. Then set the value of the bit \textit{Go} to $1$ with probability $1/2^{p_r}$ for each feedback node of the block $r$ independently.
\end{itemize}
The behavior of nodes based on their labels is the same as in the \Levelled \fastaNOS.

\begin{lemma}\label{lem:fast:z}
    If $h_2(v)=h_2(\parent(v))$ then $z(p(v))<z(v)\leq z(p(v))+7$ and $v$ receives the broadcast message by the end of block $z(v)$.
\end{lemma}
\begin{proof}
    The inequality $z(\parent(v))<z(v)$ follows directly from the definition.
    
    Assume that $h_2(v)=h_2(\parent(v)$. As $\level(v)=\level(\parent(v))+1$,
    $$y(v)-y(p(v))=6+(level(v)\mod 3)-((level(v)-1)\mod 3).$$ 
    Thus $y(p(v))<y(v)\leq y(p(v))+7$. This in turn implies that $z(v)\leq z(\parent(v))+7$.

According to the definition of $z(v)$, $\parent(v)$ receives the broadcast message and gets informed before block $z(v)$. In order to show that $v$ receives the broadcast message by the end of block $z(v)$, consider two cases corresponding to the various assignments of the bits \textit{Fast} and \textit{Rescue}:
\begin{enumerate}
\item[\textbf{Case~1.}] 
    Node $\parent(v)$  is in the dominating set $\DOM_{z(v)}$.
    %of the block $z(v)$ of the algorithm.

    Then $\parent(v)$ has some feedback node $u$ in the block $z(v)$ first informed in 
    this block and the \textit{Fast} bit of $u$ is equal $1$. The reception of the value of $1$ of the bit \textit{Fast} in the \textit{Feedback} step of the block $z(v)$ informs $\parent(v)$ that is should transmit the broadcast message in the \textit{Fast} step of the current block. Indeed,
    assume that $v$ does not receive this message from $\parent(v)$ sent in the \textit{Fast} step.
    Then, another neighbor $w$ of $v$ transmits as well, causing a collision.
    But the fact that $w$ transmits in the \textit{Fast} step of block $z(v)$ implies that $w$ is the parent of $w'$ such that $\level(w')=\level(v)$, $h_2(w')=h_2(w)$ and $z(v)=z(w')$,
    by the definition of the function $z$.
    Then, by Observation~\ref{obs:ultimate:z}, $h_2(w')=h_2(v)$.
    Therefore the edges $(v,w)$ and $(w',w)$ contradict the fact that our BFS tree $T$ is a \thrtNOS. 

\item[\textbf{Case~2.}]  
    Node $\parent(v)$ is not in the dominating set $\DOM_{z(v)}$. 
    %of the block $z(v)$ of the algorithm.

 As $v$ is still not informed before the block $z(v)$, there is a neighbor $u$ of $v$ in the dominating set $\DOM_{z(v)}$ with $\level(u) = \level(\parent(v))$. 
 Then, as $u$ is in $\DOM_{z(v)}$, it has some feedback node $u'$ first informed in 
 the block $z(v)$ such that the \textit{Rescue} bit of $u'$ is equal $1$. 
 The node $u'$ sends the value $1$ of its \textit{Rescue} bit to $u$.
 The reception of the value of $1$ of the bit \textit{Rescue} in the \textit{Feedback} step of the block $z(v)$ informs $u$ that is should transmit the broadcast message in the \textit{Rescue} step of the current block. Indeed, assume that $v$ does not receive this message from $u$ sent in the \textit{Rescue} step.
 Then, other neighbor $w$ of $v$ transmits as well, causing a collision.
    But the fact that $w$ transmits in the \textit{Rescue} step of $z(v)$ implies that 
    there is a neighbor $w'$ of $w$ such that $\level(w')=\level(v)$, $z(w')=z(v)$ and therefore, 
    by Observation~\ref{obs:ultimate:z}, $h_2(w')=h_2(v)$.
    The simultaneous existence of the edges $(v,w)$ and $(w',w)$ contradicts the fact that our BFS tree $T$ is a \thrtNOS. 
\end{enumerate}
    
\end{proof}

\begin{theorem}\label{th:fastest:alg}
The  \fastesta    runs in time $O(D + \log^2 n) $ with high probability.
\end{theorem}

\begin{proof}
%Let's take any node $u$. 
Let $u$ be an arbitrary node.
We will show that $u$ becomes informed in 
%$O(\level(v) + \log n \cdot ((h_2(r) - h_2(u))))$ 
$O(\level(v) + \log^2n)$,
rounds with probability at least $1-n^{-2}$.
This estimation combined with the union bound implies the result stated in the theorem.

%, since 
%$h_2(v)=O(\log n)$ for each node $v$ by Lemma~\ref{l:2h:max}.

\begin{comment}
Observation~\ref{obs:ultimate:z}, 
combined with the appropriate assignment of the bits \textit{Fast} and \textit{Rescue} in the labels,
and the arguments from the proof of 
 Lemma~\ref{l:alg:our:cases}
 gives us the %fact 
 conclusion that all fast and rescue transmissions will be successful without collisions on the fast tracks.    
\end{comment}
First, we introduce the auxiliary notion of a \emph{fast track}.
It is a path $(v_1,\ldots, v_k)$ in $T$ such that 
\begin{itemize}
    \item $v_i=\parent(v_{i+1})$ for each $i\in[1,k-1]$, 
    \item $h_2(v_1)=h_2(v_2)=\cdots= h_2(v_k)$ for each $i\in[1,k]$, 
    \item $v_1=s$ or $(\parent(v_1),v_1)$ is a fast edge.
    %\item $v_k$ is a leaf of $T$ or 
\end{itemize}
That is, all edges of a fast track are fast edges.
Observe that, according to Lemma~\ref{lem:fast:z}, if $(v_1,\ldots,v_k)$ is a fast track and $v_1$ receives the broadcast message in some block $b$ then $z(v_k)=b+O(k)$ and $v_k$ receives the broadcast message by the end of block $z(v_k)$. As the maximum of $h_2(v)$ over all $v\in V$ is {smaller} than $\log n$ and the value of $h_2$ cannot increase on a simple path from the root $s$ to any node $v$ of $T$, a path from $s$ to $v$ can be split into at most $\log n$ fast tracks and at most $\log n$ slow edges. As we have shown, the number of rounds needed for broadcasting a message along the fast tracks is linear with respect to the total length of these tracks which is $O(D)$. Thus, it remains to analyze the number of rounds needed to pass the broadcast message through $h_2(v)-h_2(s)<\log n$ slow edges.

Recall that a feedback node of a block $r$ is each node $v$ such that $v$ is a feedback node of some
node $v'\in\DOM_r$ in the block $r$. The random choice of the bits $Go$ assigned in labels of feedback nodes, described above, implies that, for each frontier node $w$, $w$ receives the broadcast message in the \textit{Go} step of the block $r$ with probability larger than $1/(30\log n)$, as proved in Lemma~9 in \cite{EllenG20}. 

Let $(s=v_1,\ldots, v_p=v)$ be a path in $T$ from the root $s$ to some node $v$.
Let $b_1, b_2,\ldots$ be the sequence of blocks with the property that $b_i$ is the $i$th block of an execution of our algorithm such that, at the beginning of the block $b_i$, the largest index $j$ such that $v_j$ is informed is such that $(v_j,v_{j+1})$ is a slow edge. For a fixed $i$,  let the index $j$ satisfying the properties from the previous sentence be denoted by $r_i$.
Let $X_1, X_2,\ldots$ be a sequence of independent random variables such that $X_i=1$ iff the node $v_{r_i+1}$ receives the broadcast message in the \textit{Go} step of the block $b_i$. Then $\prob(X_i=1)\geq 1/(30\log n)$ by Lemma~9 from \cite{EllenG20}. Moreover, if $\sum_{i=1}^a X_i\ge h_2(s)-h_2(v)$ then the broadcast message is already delivered through all slow edges of the path $v_1,\ldots,v_p$ until block $b_a$. 

Let $\alpha=c\cdot 60\log^2n$, let $c\ge 8$ be a constant and let $X=\sum_{i=1}^\alpha X_i$.
Then $E(X)\ge \alpha\cdot \frac1{30\log n}=2c\log n$.
If $X\ge h_2(s)-h_2(v)$ then the broadcast message is delivered through all slow edges of the path $v_1,\ldots,v_p=v$ in at most $\alpha=O(\log^2n)$ rounds.
Using standard Chernoff inequalities, we get
\begin{align*}
    \prob\left(X<h_2(s)-h_2(v)\right) &< \prob( X< \log n ) \\
    &< \prob\left(X< \frac12 E(X)\right)\\
    &\le e^{-EX/8}\le \left(\frac1{n}\right)^{2c/8}\le \frac1{n^2}    
\end{align*}
for $c\ge 8$.

\commentt{
 Thus, every node $v$ with the same 2-height value as its parent will get its message by the round $z(v)$. Let us look at the difference between the value $z(v)$ of some node $v$ and $z(\parent(v))$ of the parent of $v$ with $h_2(v) = h_2(\parent(v))$. 
 The value of $y(v)$ is larger than $y(p(v))$ by at most $6$ and at least $3$, so $z(v)$ is larger than $z(p(v))$ by at most $6$ as well. \tj{? A OPERACJE MOD NIE MOGA TEGO ZMIENIC ? }    
}
Let $v_i$ be a node in the path $v_1,\ldots,v_p$ beginning a fast track, i.e., such that $(v_{i-1},v_i)$ is a slow edge and $(v_i,v_{i+1})$ is a fast edge. Then, after reception of the broadcast message in (an arbitrary) block $r$, its ultimate round $z(v_i)$ might be by at most $6\log n$ larger than $r$.
Thus, the number of rounds lost because of slowdowns on the borders between slow edges and fast tracks is at most $6\log n(h_2(s)-h_2(v))\in O(\log^2n)$ for each node $v$, with probability at least $1-n^{-2}$.

The number of rounds required for delivery of the broadcast message to an arbitrary node $v$ is the sum of the number of rounds needed to pass the message through the fast tracks, the number of rounds needed to pass the message through slow edges and the number of rounds of slowdowns on the borders between slow edges and fast tracks. As we proved above, this sum is $O(D+\log^2n)$ with probability at least $1-1/n^2$, for each node $v$. Thus, by the union bound, the probability that any node does not receive the broadcast message within $O(D+\log^2n)$ rounds, is at most $1/n$. 
This concludes the proof.
\end{proof}

Note that the only random ingredient is in the labeling scheme assignment. Given a labeling scheme, the  \fastesta is a deterministic broadcasting algorithm. Hence Theorem \ref{th:fastest:alg} implies the following corollary.

\begin{corollary}
There exists a constant length labeling scheme supporting broadcast in time $O(D+\log^2n)$.
\end{corollary}

In view of the lower bound from \cite{ABLP}, our broadcasting time $O(D+\log^2n)$ is optimal, even when compared to broadcasting time in radio networks of known topology.

\subsection{Acknowledged broadcasting}

In this section we consider a communication task slightly more demanding than broadcasting. It is called {\em acknowledged broadcasting}.
In acknowledged broadcasting working in time $T$, we require that, not only all nodes know the broadcast message $M$ until $T$ but also each node knows the round number $T$ such that all nodes receive the broadcast message until round $T$. 

We now present a relatively simple modification of all broadcasting algorithms presented in  this paper which transforms them into acknowledged broadcasting algorithms. The proposed modified algorithms preserve time complexity of the original algorithms and extend labels by a constant length. 
%To learn the value of $\level(v)$ for each node $v$ we can to broadcast message append the value of $\level(v)$. As each node knows the value of its level modulo $3$ beforehand it can deduce its own level with this new information. 
Consider any of our broadcasting algorithms and a fixed instance of the broadcasting problem. Let $T$ be a \gt\  of the communication graph $G=(V,E)$ of that instance, with the root $s$ equal to the source node of broadcasting.
Then, choose an arbitrary node $v$ that receives the broadcast message (i.e., becomes informed) in the latest round. 
Observe that no node transmits any message after the block in which $v$ receives the broadcast message. Indeed, as only neighbors of uninformed nodes are active at the beginning of each block, there are no transmitting nodes if every node is already informed.
Let $P$ be the unique simple path in $T$ from $s$ to $v$.
Using two additional bits in the label, we encode which nodes are on the path $P$ and which node is the last node $v$ on the path. 
Directly after the block $r$ in which $v$ becomes informed, it transmits the special message \textit{Stop}. Then, each node located on the path $P$ transmits this message \textit{Stop} directly after reception of this message.
When the source gets the message \textit{Stop} in round $t_1$, it learns that $t_1$ is an upper bound on the number of rounds of the broadcast, and $t_1$ is larger than the actual broadcast time by at most the length of $P$ which is not larger than $D$.
Then the source broadcasts the value of $t_1$, and all nodes can assume that the broadcast algorithm is finished after $2t_1$ rounds.

\begin{corollary}\label{cor:broadcasting:ack}
    ~
    \begin{enumerate}
        \item  Acknowledged broadcasting in time $O(D + \min(D,\log n)\log^2 n)$ is supported by some constructive labeling scheme of constant length.
        \item There exists a labeling scheme of constant length supporting acknowledged broadcasting in time $O(D+\log^2n)$.
    \end{enumerate}
\end{corollary}

\section{Gossiping}\label{sec:k:gathering}

In this section we consider the task of gossiping. As previously announced, we first focus on the auxiliary task of gathering, in which
messages of all nodes have to be gathered in a designated node, called the {\em sink}.
Section~\ref{subsec:kgather:lower} gives a lower bound 
$\Omega(\log \Delta)$ on the length of a labeling scheme sufficient to accomplish this task.
(As gossiping is at least as hard as gathering, this is also a lower bound on the length of a labeling scheme sufficient to accomplish gossiping).
%In Section~\ref{subsec:gathering:tools}, we introduce some technical tools needed in the sequel. Section~\ref{subsec:k:gathering} is devoted to a $k$-gathering algorithm that works in time $O(D+k)$. 
In Section~\ref{subsec:gathering:Deltalog}, we present a gathering algorithm working in time $O(D+\Delta\log n+\log^2n)$. 
%Both these algorithms use labeling schemes of length $O(\min(\log\Delta,\log k))$. 
The algorithm uses a labeling scheme of asymptotically optimal length $O(\log \Delta)$.
Finally, in Section~\ref{subsec:gather:sum}, we
apply the gathering algorithm combined with our broadcasting algorithm to the gossiping problem.
%combine the results from the previous sections and apply them to the gossiping problem.

%%%%%%%%%%%%%%%%%%%%%%%%%%%%%%%%%%%%%%%%%%%%%%%%%%
\subsection{Lower bound on the length of a labeling scheme}\label{subsec:kgather:lower}
 %We will consider two cases. Firstly, the case that $k \leq \Delta$ and then $k > \Delta$.
In this section, we observe that $\Omega(\log\Delta )$ is a lower bound on the length of a labeling scheme sufficient to accomplish gathering.

Consider the graph with the set of nodes $V=\{v_1,v_2,\ldots, v_{D+\Delta}\}$,
and with the set of edges
$\{(v_i,v_{i+1})\,|\, i\in[1,D]\}\cup \{(v_{D}, v_j)\,|\, j\in[D+1,D+\Delta]\}$ (see Figure~\ref{fig:kgathering-lower}).
This graph has diameter $D$ and maximum degree $\Delta$.

\begin{figure}[h]
        \centering
        \includegraphics[width=0.7\linewidth]{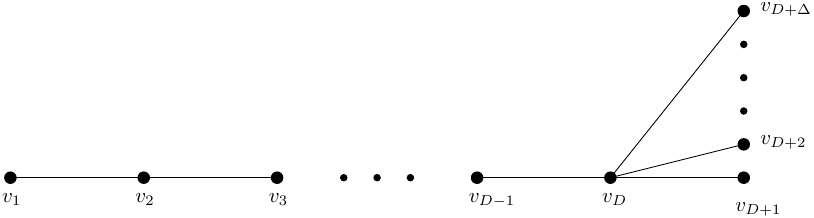}
        \caption{Illustration for the lower bounds on $k$-gathering.}
        \label{fig:kgathering-lower}
        \end{figure}    

Assume that %$p=\Delta$, $k$ input messages are originally located in the nodes $v_{D+1},\ldots,v_{D+\Delta} and 
the sink is $v_1$. As the unique path to $v_1$ from each node $v_i$, for $i>D$, goes through $v_D$, all messages from these nodes have to reach $v_D$ before reaching the sink.

%we see that at least $\Delta$ rounds are needed in order to transfer all  messages originating in $v_{D+1},\ldots,v_{D+\Delta}$  to %$v_D$. Then, additional $D-1$ rounds are necessary to pass any message from $v_D$ to $v_1$. These observations imply the lower bound %$D+\Delta -1$ on the number of rounds for the gathering problem in the considered graph.

%For a lower bound on the length of the labeling scheme, consider again the graph from Figure~\ref{fig:kgathering-lower}.
%, possibly extended by some subgraphs connected with some of $v_1,\ldots,v_{D-1}$ by a single edge.
%Firstly assume that $p=k$, the source messages are originally located in the nodes $v_{D+1},\ldots,v_{D+k}$ and they should be gathered at $v_1$.
All nodes $v_{D+1},\ldots,v_{D+\Delta}$ must have distinct labels in order to deliver their messages to $v_D$ (and then to $v_1$),
otherwise there is a collision at $v_D$. These distinct labels require a labeling scheme of length $\Omega(\log \Delta)$.
%Secondly, if $p=\Delta<k$ and $\Delta$ of the $k$ source messages are located originally in $v_{D+1},\ldots,v_{D+\Delta}$, the labels of size $\log \Delta$ are necessary.
            %%%%%%%%%%%%%%%%%%%%%%%%%%%%%%%%%%%%%%%%%%%%%%%%%%%%%
Hence we have:

\begin{proposition}\label{th:kgathering:lower}
    The gathering task requires a labeling scheme of length
     $\Omega(\log \Delta)$.
\end{proposition}
\commentt{
Let us note here that the above lower bound on the size of labels could be also obtained by a reduction of the $k$-broadcasting problem from \cite{DBLP:conf/wg/KriskoM21} to a composition a $k$-gathering and standard broadcasting and the fact that broadcasting can be accomplished with labels of constant size. The lower bound on $k$-broadcasting is provided in \cite{DBLP:conf/wg/KriskoM21}.
}

\subsection{Gathering in time $O(D + \Delta \log n+\log^2n )$ with 
%labels $O(\log \Delta)$}
optimal length of labels}\label{subsec:gathering:Deltalog}

%\begin{center}
%	\textbf{\LARGE{BELOW IS THE TEXT EDITED CURRENTLY}}
%\end{center}

%\textbf{komentarz:}
%poniższy algorytm pomimo bycia algorytmem na gathering korzysta z definicji \gt.
 
In this section we describe an algorithm for gathering, using a labeling scheme of optimal length $O(\log \Delta)$ and running in time $O\left(  D + \Delta \log n+\log^2n \right)$. 
%Thus, the algorithm gives a better time bound than the $D+k$ algorithm from Theorem~\ref{t:k-gather} for $\Delta\log n=o(k)$. 

Our algorithm makes use of some ideas of the centralized gathering algorithm from \cite{DBLP:journals/dc/GasieniecPX07}. However, as in the case of broadcasting,
adjusting a centralized algorithm to the regime of distributed algorithms with short labels
requires some significant changes in the original centralized algorithm. 
In particular, we use the notion of a \gt\ which strengthens properties of a \emph{gathering-broadcasting spanning tree} from \cite{DBLP:journals/dc/GasieniecPX07}.
%Interestingly, we identify a kind of optimization which makes our solution to work around twice faster than the offline algorithm from %\cite{DBLP:journals/dc/GasieniecPX07}, although asymptotic complexity of both algorithms are equal. Moreover, our algorithm is energy %efficient, assuming that energy consumption is determined by the number of transmissions of a node. Indeed, each node will transmit %only once in our algorithm, while $\Omega(\Delta)$ transmissions of a node are necessary in the solution from \cite{DBLP:journals/dc/%GasieniecPX07}.

We will use a \gt\ as a backbone for transmissions of messages,  aiming at gathering all messages in the sink node. In particular, we will use the following observation.
\begin{observation}\label{obs:thrt}
	If nodes %$u$ and $v$ %, 
    $u \neq v$
    are such that $\level(u)=\level(v)$ and $h_2(u)=h_2(v)=h_2(\parent(u))=h_2(\parent(v))$ in some \gt\  $T$, then $u$ and $v$ can simultaneously send messages to their parents in $T$ without a collision.
\end{observation}

We now describe a centralized algorithm GatherCentr for gathering in at most $3D+6(\Delta+1) \log n$ rounds. Then we provide a labeling scheme of length $O(\log \Delta)$ and the distributed algorithm GatherDistr using it, that simulates the centralized algorithm.
The time complexity of algorithm GatherDistr is $O(D+\Delta\log n+ \log^3n)$ for the constructive variant of labeling and $O(D+\Delta\log n+\log^2n)$ in the case of the labeling scheme which is not obtained constructively.
%without loss in time complexity.
%
It should be stressed that, while GatherCentr follows some ideas of the centralized gathering algorithm from \cite{DBLP:journals/dc/GasieniecPX07}, its analysis is different, as it makes use of the notion of a \gt. 
%as well as provides smaller multiplicative constants ($3D$ instead of $6D$) hidden in asymptotic time complexity expressions.  
%
%In the offline algorithm by G\k{a}sieniec et al at first we fix \textit{gathering-broadcasting} spanning Tree -- similar to our \gt. 

Let $T$ be a BFS which is also a \thrt rooted at the sink vertex $s$.
Recall that such a tree can be constructed in polynomial time, by Lemma~\ref{lem:2hrt}.
Then, we split the set of nodes $V$ into sets 
$F$ and $S$, where $v \in F$ if $h_2(v)=h_2(\parent(v))$, and $v \in S$  otherwise. 
This partition might be also seen as a partition into the set of  nodes $F$ 
%which can transmit messages to 
connected with their parents through \emph{fast} edges and the set $S$ which contains the nodes %transmitting messages to 
connected with their parents through \emph{slow edges}. (As in the broadcasting algorithms, an edge $(u,p(u))$ of $T$ is a \emph{fast edge} if $h_2(u)=h_2(p(u))$. Otherwise, an edge is a \emph{slow edge}). Note that, by Observation~\ref{obs:thrt}, all nodes in $F$ from a given level and with the same value of their 2-heights can simultaneously transmit messages to their parents without any collision. %This property complements transmissions from nodes to their children through fast edges exploited in our algorithms.
%
%Lemma 21 in \cite{DBLP:conf/wdag/GanczorzJLP21} 

Lemma 6 from \cite{DBLP:journals/iandc/GanczorzJLP23}
states that there is an efficient way to assign the value $s(v) \in [0, \Delta-1]$ to each node $v$ of any BFS tree $T$ so that 
\begin{itemize}
    \item $s(u)\neq s(w)$ for different children $u,w$ of a node of $v$ in $T$,
%for children of a node  have distinct value of $s$ and 
    \item there is no edge between $u$ and $\parent(w)$ for any two nodes $u\neq w$ on the same level of $T$ such that $s(u)=s(w)$. 
\end{itemize}

\noindent\textbf{Centralized gathering}\\
Algorithm GatherCentr works as follows. First, we group time steps in blocks of length $3$. A node $v$ with $\level(v) \mod 3 = i$ can transmit only in the $i$th step of a block for $i\in[0,2]$, provided that we count the steps of a block starting from zero. 
We will define the  transmitting block $t(v)$ for each node $v$ as follows using the value $s(v)$ of the node $v$:
\begin{enumerate}
	\item if $v \in F$, then $t(v)= (D - \level(v)) + h_2(v) \cdot (\Delta+1)$
	\item if $v \in S$, then $t(v)= (D - \level(v)) + h_2(v) \cdot (\Delta+1) + s(v) + 1$
\end{enumerate}
Each node $v$ listens until the end of block $t(v)-1$ and then it transmits all gathered messages along with its own message during the block $t(v)$.

In the following lemma we prove by induction on $t(v)$ that each node receives all messages from its subtree of $T$ before the block $t(v)$.
However, for a better understanding of the idea behind the algorithm, we first provide some intuitions hidden in the formal proof. 
%Let $P$ be a subpath with respect to $F$ of nodes a path from a node $v$ to the root $s$ of $T$. Here a subpath $P$ is maximal with respect to $F$ if
%\begin{itemize}
%    \item all edges from $P$ are fast, 
%    \item the first vertex of $P$ is a leaf or all edges
%\end{itemize}
Let $P$ be a subpath of a path going from a leaf to the root $s$ of $T$, such that all edges connecting nodes of $P$ are fast. Let $(v_1,\ldots,v_p)$ be the sequence of nodes of $P$ starting from the node on the largest level. 
As all the edges on $P$ are fast, we have $h_2(v_1)=h_2(v_2)=\cdots=h_2(v_p)$, 
and the nodes $v_1,\ldots,v_{p-1}$ belong to $F$. Thus, $t(v_{i+1})=t(v_i)+1$, for $i<p$, and the message transmitted by $v_i$ in block $t(v_i)$ is received by the end of $t(v_{i+1})$. 
By Observation~\ref{obs:thrt}, a message from $v_i$ to $v_{i+1}$ transmitted in block $t(v_i)$ is received without a collision in the block $t(v_i)$, since each other transmitter $v'$ from the level $\level(v_i)$ in that block
is such that $h_2(v')=h_2(v_i)$. 
Apart from Observation~\ref{obs:thrt}, collision-free transmissions through fast edges are guaranteed thanks to the fact that, for each level $l$, transmissions from that level through fast edges do not interfere with transmissions through slow edges. Indeed, $t(v)-(D-l)$ is divisible by $\Delta+1$ iff $v$ is connected with its parent by a fast edge.
Thus, transmission of messages through such fast path $P$ of length $p$ takes $p$ blocks. That is, the time of this transmission is proportional to the length of $P$. Moreover, each path from a leaf to the sink $s$ of $T$ can be split into at most $\log n$ fast paths separated by at most $\log n$ slow edges, due to Lemma~\ref{l:2h:max}.  
%As the level of consecutive nodes on such a path decreases by one with each edge, 
%
Additionally, observe that $t(v)< t(u)\le t(v)+\Delta+1$ for an edge connecting $u\in S$ with $v\in F$ or $u\in F$ with $v\in S$ and $\level(u)=\level(v)+1$.
Thus, each edge between a node $v\in F$ and a node $w\in S$ on a path from a leaf to the sink gives a slowdown of at most $\Delta$ blocks. As there are at most $\log n$ slow edges on such a path, this slowdown is at most $\Delta \log n$.
\begin{lemma}\label{lem:kgather:large:k:induction}
	%A 
    Each
    node $v$ gets all messages from its subtree before %round $T(v)$.
    the block $t(v)$.
\end{lemma}

\begin{proof}
We will prove the lemma by induction on $t(v)$, for $v\in V$. In the base case $t(v)=1$ we have only such nodes $v$ that $D=\level(v)$. That is $t(v)=1$ only for leaves
of $T$. The only message in the subtree of such a node is its own message. 

For the inductive step, assume that the lemma is satisfied for each node $u$ such that $t(u)<X$ for some $X>1$. Let %us fix some value $X>0$ and a node 
$v$ be a node with $t(v) = X$.
First, observe that $t(u)$ is smaller than $t(v)$ for each child $u$ of $v$.  By the inductive hypothesis, $u$ got the messages from its subtree before $t(u)$, for each child $u$ of $v$ in $T$. 
Thus, it remains to show that children of $v$ successfully transmit their knowledge to $v$ without collisions before $t(v)$. 

As $T$ is a BFS tree and all nodes $w$ with $\level(w) \mod 3 = i$ transmit in the $i$th step of a block for $i\in[0,2]$, there are no collisions between messages transmitted by nodes on different levels. Let $u$ be a child of $v$ and assume, for a contradiction, that the transmission of $u$ in the block $t(u)$ collides at $v$ with a message of some node $u'$ transmitted in the same round. Then  $\level(u)=\level(u')$, due to the partition of each block in three rounds corresponding to transmissions of nodes with various values of levels mod $3$.
As 
$s(u),s(u') <  \Delta$,
$\level(u)=\level(u')$ and the collision of $u$ with $u'$ is only possible when $t(u)=t(u')$, it is sufficient to consider the following cases: 
\begin{itemize}
    \item $h_2(u)\neq h_2(u')$.
    
    Then, according to the definition of the transmitting round of a node,  $t(u)\neq t(u')$ by the fact that $s(u), s(u')<\Delta$ and thus $u$ and $u'$ cannot collide -- we get a contradiction.

    \item $u,u'\in S$ and $h_2(u)=h_2(u')$.

    If $s(u)\neq s(u')$ then $t(u)\neq t(u')$ which contradicts the assumption $t(u)=t(u')$.

    Thus, it remains to consider the case that $s(u)=s(u')$.
    As $u$ and $u'$ have a common neighbor $v$ on the smaller level $l-1$, the equality $s(u)=s(u')$ contradicts the properties of the assignment of values $s(w)$ for $w\in V$ assured by Lemma 6 from \cite{DBLP:journals/iandc/GanczorzJLP23}.

    \item $u\in F$ or $u'\in F$, $h_2(u)=h_2(u')$.

    If both $u\in F$ and $u'\in F$ then the fact that $u$ and $u'$ have a common neighbor $v$ on the smaller level $l-1$ contradicts Observation~\ref{obs:thrt}.
    If only one of $u,v$ belongs to $F$, assume w.l.o.g.\ that $u\in F$ and $v\in S$. Then $t(u)\neq t(u')$ by the definition, since $t(u)-(D-\level(u))$ is divisible by $\Delta +1$ and $t(v)$ is not divisible by $\Delta+1$. Thus transmissions of $u$ and $u'$ do not collide.
\end{itemize}

%also $h_2(u)=h_2(u')$ and $s(u)=s(u')$. But if the values of $s(u)$ and $s(u')$ are equal, there is no edge between $u'$ and $\parent(u)=v$ by the definition of the assignment of the values $s(v)$ to all the nodes $v\in V$. But this lack of an edge $(u',\parent(u))$ contradicts the contrary assumption that $u'$ causes a collision of the transmission of $u$ at $\parent(u)$. 
%This fact finishes the proof of the lemma. 
%
The above inspection of all possible cases concludes the proof of the lemma.
\end{proof}
Lemma~\ref{lem:kgather:large:k:induction} combined with the definition of $t(v)$ implies the following corollary.
\begin{corollary}\label{cor:GatherOffLargek}
    Algorithm GatherCentr finishes gathering in at most $3D+6\Delta \ceil{\log  n}$ rounds.
\end{corollary}

%\subsection{Encoding of necessary information}

\vspace*{1em}
\noindent\textbf{Labeling scheme and the distributed gathering algorithm.}\\
\noindent We now show how to learn all information needed for a node to simulate the centralized gathering algorithm described above, using a labeling scheme of length $O(\log \Delta)$. First observe that the value $\Delta$, $s(v)$ and the bit indicating whether a node belongs to $F$ or to $S$ can be encoded in $O(\log \Delta)$ bits. 
However, in order to determine the block $t(v)$ of transmission of a message by $v$, the values of $D$ and $\level(v)$ are needed.
In order to assure that a node $v$ for each $v\in V$ can learn the values of $D$ and $\level(v)$, a simulation of the centralized algorithm GatherCentr will be preceded by:
\begin{itemize}
    \item An execution of the \textit{Size Learning} algorithm from \cite{DBLP:conf/wdag/GanczorzJLP21} in order for all nodes to learn the value of $D$.
    This is possible since the \textit{Size Learning} algorithm is in fact an algorithm for learning any message of size $O(\log n)$ in $O(\log^2n)$ rounds with labels of length at most $O(\log \Delta)$.

    To make it possible, the labels for the \textit{Size Learning} algorithm will be a part of the labels in our gathering algorithm.
    
    \item An execution of the appropriately modified acknowlwedged broadcasting algorithm (see Corollary~\ref{cor:broadcasting:ack}) with the source node equal to the sink node of the considered instance of gathering, in order to assure that each node $v\in V$ learns the value of $\level(v)$. 

    During an execution of broadcasting, the source will transmit the message that its level is equal to $0$ together with the broadcast message. Then, each node which knows its level $l$ will send $l$ together with the broadcast message. Given the fact that labels of our broadcasting algorithm contain the values of their levels mod $3$, all nodes can learn their levels upon reception of the broadcast message by them, since each node $v$ can filter out messages from nodes on the levels $\level(v)$ and $\level(v)+1$ and thus inherit its level from the level of a nodes at the preceding level that successfully delivers the broadcast message to $v$.
\end{itemize}
Moreover, we use one extra bit of labels to mark the leaves of the tree which also provides information to each leaf $v$ that $h_2(v)=0$.  
Finally, each non-leaf node $v\in V$ should somehow learn the value of $h_2(v)$ during an execution of the gathering algorithm early enough, i.e., before the block $t(v)$. 
%\tju{SZCZEGOLY LATER?}
%run the special acknowledged broadcast algorithm, and the value of $h_2(v)$ will be learned during the gathering phase.
The key obstacle for achieving it is the fact that, as our goal is to keep the length $O(\log \Delta)$ of labels while the execution time of the broadcasting algorithm and {\em Size Learning} algorithm depend also on $D$ and $n$, it is challenging to synchronize the nodes so that they start consecutive phases of these auxiliary protocols at the same time, and even more importantly, start an execution of the actual algorithm GatherDistr simultaneously. 

To this end, we start with an execution of the {\em Size Learning} algorithm whose goal is to provide the value of $D$ to all nodes. Using one additional bit in the labels, we mark exactly one node $v$ which transmits a message during the execution of {\em Size Learning} in the latest round. After its last transmission, $v$ starts an execution of the acknowledged broadcast algorithm with the broadcast message containing the number of the round $\tau$ in which the gathering algorithm should start.

Note however that we have not yet described the way in which the nodes learn their values of the function $h_2$, and the values of $h_2$ are needed to determine the block number $t(v)$ of transmissions of each $v\in V$.

%\vspace*{1em}
%\noindent\textbf{$k$-gathering distributed algorithm}
%\vspace*{1em}

%In this section 

Therefore now we will describe how to extend the labels in order to assure that each $v\in V$ can learn
$h_2(v)$ before block $t(v)$, provided it
 knows the values of $\Delta, D, \level(v),s(v)$ and knows whether it belongs to $S$ or $F$. 
In order to facilitate learning of $h_2(v)$, for each node $v$, before block $t(v)$,
%For that 
we add (binary representations of) numbers $s'(v)$ and $b(v)$ to the label of each node $v$, defined as follows.
If $v$ is a leaf then $s'(v)=b(v)=-1$.
Otherwise, let $u$ be a child of $v$ with the largest $2-$height among the children of $v$. %if $u$ is not a leaf and $u$ is equal to $-1$ otherwise.
Then, $s'(v)$ is equal to the value of $s(u)$, %for the child $u$ of $v$ , 
and $b(v)\in\{0,1\}$ is such that $h_2(v) = h_2(u) + b(v)$. 

After these preparations,
all nodes start algorithm GatherDistr in round $\tau$ by
running the centralized algorithm GatherCentr with the following modifications.
\begin{enumerate}
	\item Let $M_v$ be the set of all messages received by $v$ from its subtree, including the message of $v$ itself. Instead of transmitting $M_v$, $v$ transmits the tuple $(M_v, h_2(v), s(v), \level(v))$ in the appropriate round of block $t(v)$.
	\item If $s'(v) = -1$ then $v$ sets the value of its own $2-$height to $0$.
	\item If $v$ receives the message $(M, h, s, l)$ such that $s=s'(v)$ and $\level(v) = l-1$, it sets $h_2(v) = h + b(v)$.
\end{enumerate}

\begin{lemma}\label{lem:gather:largek:distr}
	Each node $v$ correctly determines its $2-$height before its transmitting block $t(v)$.
\end{lemma}
\begin{proof}
First, observe that no node $v$ deduces an incorrect value of $h_2(v)$, 
provided that its children determined their values of $h_2$ correctly. 
 Indeed, as by the definition of $s(v)$, there is at most one neighbor $u$ of $v$ such that $s(u) = s'(v)$ and $\level(u)=\level(v) -1$, $v$ correctly determines its $2$-height, provided it receives the correct value of $2$-heights of its children.

    It remains to show that each node $v$ learns the value of $h_2(v)$ before the block $t(v)$.
	%As for the time of learning the $2-$height 
    We will prove this fact %we can proceed 
    by induction on $t(v)$. For the base case $t(v)=0$, observe that $t(v)=0$ only if $v$ is a leaf and $v\in F$. As each node can determine this information from its label (in particular, $v$ is a leaf iff $s'(v)=-1$), the lemma holds for each node $v$ such that $t(v)=0$. 
    
    %$v$ correctly determine its value by checking that $s'(v)=-1$ in round $0$. 
    For the inductive step, assume that the lemma holds for each $v$ such that $t(v)<t$ for some $t>0$. Assume that $t(v)=t$.
    By the inductive hypothesis and by the correctness of the centralized algorithm, all children of $v$ transmit collision-free their messages to $v$ before the block $t(v)$. In particular, the only node $u$ such that $s(u)=s'(v)$ successfully transmits before block $t(v)$. So $v$ can determine $h_2(v)$ from the message received from $u$ before block $t(v)$.
\end{proof}
Using Lemma~\ref{lem:gather:largek:distr}, the correctness and complexity of algorithm GatherCentr and the complexity of our acknowledged broadcasting algorithms we get the following theorem.

\begin{theorem}\label{th:gather:largek}\begin{enumerate}
    \item 
    Algorithm GatherDistr is a distributed algorithm for gathering, working in {time $O(D + \Delta\log n+\min(D,\log n)\log^2n)$,} using a  constructive labeling scheme of length $O(\log\Delta)$.
    \item 
    There exists a labeling scheme of length $O(\log\Delta)$ such that algorithm GatherDistr using it  accomplishes 
    gathering in time $O(D + \Delta\log n+\log^2n)$.    
\end{enumerate}
\end{theorem}
%
%This with the Broadcasting algorithm gives us.

%\subsubsection{$D\Delta$ Algorithm }

%Combining the results in this section we get the proof for the Theorem \ref{t:k-gather}.

%\subsection{Combining various bounds for $k$-gathering and its consequences}\label{subsec:gather:sum}
\subsection{Application to gossiping}\label{subsec:gather:sum}

\commentt{
Yet another bound on the time for $k$-gathering directly follows from 
Lemma~7 in \cite{DBLP:conf/wdag/GanczorzJLP21}.
\begin{lemma}\cite{DBLP:conf/wdag/GanczorzJLP21}
	There is an algorithm for gathering working in time $O(D \Delta)$ with labels of size $O(\log \Delta)$
\end{lemma}
As the optimal size of labels for $k$-gathering is $\Omega\left(\min(\log k, \log \Delta)\right)$,
%\cite{DBLP:conf/wg/KriskoM21}, 
we obtain the following results.

\begin{corollary}
\begin{enumerate}
    \item 
    There exists a distributed algorithm for $k$-gathering in time $$O(D + \Delta\log n + \log^3 n)$$ in ad-hoc Radio Network Model with \textbf{efficient} labeling scheme of optimal length $O(\log\Delta)$.
    \item 
    There exists a distributed algorithm for $k$-gathering in time $$O(D + \Delta\log n + \log^2 n)$$ in ad-hoc Radio Network Model with labeling scheme of optimal length $O(\log\Delta)$.
    \end{enumerate}
\end{corollary}
}

%As the gossiping problem corresponds to an instance of gathering problem for $k=n$, we obtain the following bound on complexity and the size of labels for the gossiping problem.
Observe that the gossiping problem can be solved by an execution of a gathering algorithm for an arbitrarily chosen sink node $s$, followed by an execution of a broadcasting algorithm with the source node $s$. The only difficulty with such a composition is to synchronize both executions, so that they do not interfere and the latter one starts without significant delay. Note however that,
\begin{itemize}
    \item each node transmits a message exactly once during an execution of our gathering algorithm;
    \item one can store the degree of the sink node $s$ in the label of $s$, using extra $O(\log \Delta)$ bits.
\end{itemize}
Given the above observations it is clear that the sink node $s$ can be aware of the round in which all messages are already delivered to $s$. After this round, $s$ can start an execution of a broadcasting algorithm with the broadcast message consisting of the messages of all nodes. Consequently, our main result concerning gossiping is implied by Theorems~\ref{th:gather:largek}, \ref{th:fast:alg} and \ref{th:fastest:alg}

\begin{theorem}\label{C:gossip}
%	There exists an algorithm for gossiping in time $O(D + \Delta\log n + \log^2 n)$ in Radio Network Model with labels of length $O(\log\Delta)$.

\begin{enumerate}
    
\item 
There exists a distributed algorithm for gossiping in {time $O(D + \Delta\log n + \min(D,\log n)\log^2 n)$,} using a constructive labeling scheme of optimal length $O(\log\Delta)$.

 \item 
There exists a distributed algorithm for gossiping in time $O(D + \Delta\log n + \log^2 n)$, using some labeling scheme of optimal length $O(\log\Delta)$.
    \end{enumerate}
\end{theorem}
It should be stressed that the second time bound matches the time of the fastest known {\em centralized} gossiping algorithm. \cite{DBLP:journals/dc/GasieniecPX07}.
\commentt{
By combining Corollary \ref{C:gossip} with the results regarding the $k-$gathering we get the following corollaries.
\begin{corollary}
	There exists an algorithm for $k-$gathering in time $O(\min(n, D+k, D + \Delta\log n + \log^2 n)$ in Radio Network Model with optimal labels' length $O(\min(\log\Delta, \log k))$.
\end{corollary}
\begin{corollary}
	There exists an algorithm for $k-$broadcast in time $O(\min(n +\log^2 n, D+k +\log^2 n, D + \Delta\log n + \log^2 n)$ in Radio Network Model with optimal labels'  length $O(\min(\log\Delta, \log k))$.
\end{corollary}
}

\section{Conclusion and Open Problems}\label{sec:summary}
We presented distributed algorithms for the tasks of broadcasting and gossiping, which use labeling schemes of optimal length. In the case of broadcasting, this optimal length of a labeling scheme is constant and the time is optimal, even when compared to algorithms knowing the topology of the graph. For gossiping, our distributed algorithm uses a labeling scheme of optimal length $O(\log \Delta)$, and runs in the best known time for gossiping, even among algorithms knowing the topology of the graph.

Our results yield two interesting problems concerning the above communication tasks. The first problem concerns broadcasting. Is it possible to provide a constructive labeling scheme of constant length that supports broadcasting in time $O(D+\log^2n)$? (Our solution uses a non-constructive labeling scheme to get this optimal broadcasting time).

The second problem concerns gossiping. What is the time of the fastest gossiping algorithm using a labeling scheme of optimal length 
$O(\log \Delta)$, and does there exist a gossiping algorithm running in this time and using a constructive labeling scheme of optimal length? Note that, if our gossiping time could be improved, this would imply improving the best known gossiping time for centralized algorithms, i.e., those knowing the topology of the graph.

\bibliographystyle{abbrv}
\bibliography{references}

\end{document}